\documentclass[3p]{elsarticle}

\usepackage[english]{babel}
\usepackage[latin1]{inputenc} 
\usepackage[T1]{fontenc}
\usepackage{epsfig}
\usepackage{graphicx}
\usepackage{amssymb}
\usepackage{amsfonts}
\usepackage{amsmath}
\usepackage{verbatim}
\usepackage{url}
\usepackage[usenames]{color}
\usepackage{gastex}
\usepackage{pstricks}
\usepackage{floatflt}
\usepackage{stmaryrd}

\newcommand{\POSentiers}{\mathbb{N} \setminus \{0\}}
\newcommand{\fo}[2]{{\rm FO}^{#1}(#2)}
\newcommand{\fon}[2]{{\rm FO}_{#2}^{#1}}
\newcommand{\fLTL}[2]{{\rm LTL}^{#1}_{#2}}
\newcommand{\fMC}[2]{{\rm MC(LTL)}^{#1}_{#2}}
\newcommand{\foMC}[2]{{\rm MC(FO)}^{#1}_{#2}}
\newcommand{\fotwoMC}[2]{{\rm MC(FO)}^{#1}_{2}}
\newcommand{\fpureMC}[2]{{\rm PureMC(LTL)}^{#1}_{#2}}
\newcommand{\fopureMC}[2]{{\rm PureMC(FO)}^{#1}_{#2}}

\newcommand{\defeq}{\overset{\textsf{def}}{\Leftrightarrow}}

\newcommand{\set}[1]{\{ #1 \}}
\newcommand{\bigset}[1]{\big\{ #1 \big\}}
\newcommand{\interval}[2]{\{ #1, \ldots,#2 \}}

\newcommand{\pair}[2]{\langle #1,#2 \rangle}
\newcommand{\triple}[3]{\langle #1,#2,#3 \rangle}
\newcommand{\quadruple}[4]{\langle #1,#2,#3,#4 \rangle}
\newcommand{\Nat}{\mathbb{N}}

\newcommand{\domain}[1]{{\rm dom}(#1)}

\newcommand{\length}[1]{|#1|}

\newcommand{\avarprop}{{\rm p}}

\newcommand{\aformula}{\phi}
\newcommand{\aformulabis}{\psi}
\newcommand{\aformulater}{\varphi}
\newcommand{\mynext}{\mathtt{X}}

\newcommand{\until}{\mathtt{U}}

\newcommand{\sometimes}{\mathtt{F}}

\newcommand{\always}{\mathtt{G}}

\newcommand{\ptime}{\textsc{PTime}}

\newcommand{\pspace}{\textsc{PSpace}}

\newcommand{\egdef}{\stackrel{\mbox{\begin{tiny}def\end{tiny}}}{=}}

\newcommand{\locs}{Q}
\newcommand{\aloc}{q}
\newcommand{\alocbis}{p}

\newcommand{\aconf}{c}

\newcommand{\aautomaton}{\mathcal{A}}
\newcommand{\aautomatonbis}{\mathcal{B}}

\newcommand{\step}[1]{\xrightarrow{\!\!#1\!\!}}

\newcommand{\arun}{\rho} 

\newcommand{\infin}{\omega}
\newcommand{\fin}{*}
\newcommand{\aalphabet}{\Sigma}
\newcommand{\finwords}{\aalphabet^*}
\newcommand{\infwords}{\aalphabet^\omega}
\newcommand{\finandinfwords}{\aalphabet^\infty}
\newcommand{\aletter}{a}

\newcommand{\adataword}{\sigma}
\newcommand{\aregval}{v}
\newcommand{\avarval}{u}

\newcommand{\amap}{f}
\newcommand{\atranslation}{{\rm T}}

\newcommand{\QBF}{\textsf{QBF}}

\newcommand{\avariable}{\mathtt{x}}
\newcommand{\avariablebis}{\mathtt{y}}
\newcommand{\avariableter}{\mathtt{z}}

\newcommand{\atransition}{t}

\newcommand{\aregister}{r}

\newcommand{\inc}{\mathtt{inc}}
\newcommand{\dec}{\mathtt{dec}}
\newcommand{\ifzero}{\mathtt{ifzero}}

\newif\iffossacs
\fossacsfalse

\newif\ifptime
\ptimefalse

{\itshape}{\rmfamily}
\newtheorem{corollary}{Corollary}{\bfseries}{\itshape}
{\bfseries}{\itshape}
\newtheorem{example}{Example}
{\itshape}{\rmfamily}
\newtheorem{theorem}{Theorem}{\bfseries}{\itshape}
\newtheorem{lemma}{Lemma}{\bfseries}{\itshape}
{\itshape}{\rmfamily}
{\itshape}{\rmfamily}
\newtheorem{proposition}{Proposition}{\bfseries}{\itshape}
{\itshape}{\rmfamily}
{\itshape}{\rmfamily}
{\itshape}{\rmfamily}
\newenvironment{proof}{\par\textit{Proof.}}

\makeatletter
\let\c@definition\c@theorem
\let\c@lemma\c@theorem
\let\c@corollary\c@theorem
\let\c@proposition\c@theorem
\makeatother

\begin{document}
\begin{frontmatter}
\title{Model checking  memoryful linear-time logics over one-counter 
automata\tnoteref{averiss}}
\tnotetext[averiss]{Supported by the Agence Nationale de la Recherche, grant ANR-06-SETIN-001.}

\author[LSV]{St\'ephane Demri}
\address[LSV]{LSV, ENS Cachan, CNRS, INRIA Saclay IdF, France
% \\ {\small email: \protect\url{demri@lsv.ens-cachan.fr}}
}
\author[Warwick]{Ranko Lazi{\'c}}
\address[Warwick]{Department of Computer Science, University of Warwick, UK
%\\ {\small email: \protect\url{lazic@dcs.warwick.ac.uk}}
}
\author[EDF]{Arnaud Sangnier}
\address[EDF]{LSV, ENS Cachan, CNRS \& EDF R\&D, France
%\\ {\small email: \protect\url{sangnier@lsv.ens-cachan.fr}}
}

%\maketitle

\begin{abstract}

We study  complexity of the model-checking problems for LTL with registers (also known as freeze LTL and written $\fLTL{\downarrow}{}$)
and for first-order logic with data equality tests (written $\fo{}{\sim,<,+1}$)
over one-counter automata. We consider several
classes of one-counter
automata (mainly deterministic vs. nondeterministic)
and
several logical  fragments (restriction on the number of registers or variables and on the use of
propositional variables for control states). The logics have the ability to store a counter value and
to test it later against the current counter value.
We show that model checking  $\fLTL{\downarrow}{}$ and $\fo{}{\sim,<,+1}$ over deterministic one-counter automata 
 is  \pspace-complete with infinite and finite accepting runs.  
By constrast, we prove that model checking   $\fLTL{\downarrow}{}$  in which the until operator 
$\until$ is restricted to the eventually $\sometimes$
over nondeterministic one-counter automata  is $\Sigma_1^1$-complete 
[resp.\ $\Sigma_1^0$-complete] in the infinitary [resp.\ finitary] case
even if only one register is used and with no propositional variable.
As a corollary of our proof, this also holds for $\fo{}{\sim,<,+1}$ restricted to two variables (written 
$\fon{}{2}(\sim,<,+1)$). 
This makes a difference with the facts that  several verification problems for one-counter
      automata are known to be decidable with relatively 
      low complexity, and that finitary satisfiability for 
      $\fLTL{\downarrow}{}$ and $\fon{}{2}(\sim,<,+1)$ are decidable. 
Our results pave the way
for  model-checking  memoryful (linear-time) logics
over other classes of operational models, such as reversal-bounded counter machines.
\end{abstract}

\end{frontmatter}

\section{Introduction}

\textit{Logics for data words.} Data words are sequences in which each position is labelled by a letter  
from a finite alphabet and by another letter from an infinite alphabet (the datum). 
This fundamental and simple model arises in systems that are potentially unbounded in some way.
Typical examples  are
runs of counter systems~\cite{Minsky67}, 
timed words accepted by timed automata~\cite{Alur&Dill94}
and runs of systems with unboundedly many parallel components (data are component indices)~\cite{Bjorklund&Bojanczyk07b}.
The extension to trees makes also sense to model XML documents with values, 
see e.g.~\cite{Bojanczyketal09,Bjorklund&Bojanczyk07,Jurdzinski&Lazic07}. 
In order to really speak about data, known logical formalisms for 
data words/trees contain a mechanism that stores
a value and tests it later against other values, 
see e.g.~\cite{Bojanczyketal06a,Demri&Lazic09}.
This is a powerful feature shared by other 
memoryful temporal
logics~\cite{Laroussinie&Markey&Schnoebelen02,Kupferman&Vardi06}.
However, the satisfiability problem for these logics becomes easily undecidable even when
stored data can be tested only for equality.
For instance,
first-order logic for data words restricted to three individual variables is 
undecidable~\cite{Bojanczyketal06a} and LTL with registers (also known as freeze LTL) 
restricted
to a single register is  undecidable over infinite
data words~\cite{Demri&Lazic09}. By contrast, decidable fragments of the satisfiability problems
have been found in~\cite{David04,Bojanczyketal06a,Demri&Lazic&Nowak06,Demri&Lazic09,Lazic06} either by imposing
syntactic restrictions (bound the number of registers, constrain the
polarity of temporal formulae, etc.) or 
by considering subclasses of data words (finiteness for example). 
Similar phenomena occur with metric temporal logics and 
timed words~\cite{Ouaknine&Worrell06a,Ouaknine&Worrell07}.
A key point for all these logical formalisms is the ability to store a value from an infinite
alphabet, which is a feature also present in models of register 
automata, see e.g.~\cite{Bouyer&Petit&Therien03,Neven&Schwentick&Vianu04,Segoufin06,Bjorklund&Schwentick07}.
However, the storing mechanism has a long tradition (apart from its ubiquity in
programming languages) since
 it appeared for instance in real-time
logics~\cite{Alur&Henzinger89} (the data are time values) and in so-called hybrid 
logics (the data are node addresses), see an early undecidability
result with reference pointers in~\cite{Goranko96}. Meaningful restrictions for hybrid logics
can also lead to decidable fragments, see e.g.~\cite{Schwentick&Weber07}.
%
% Memoryful temporal logics dedicated to model-checking 
% tasks have been also investigated in~\cite{Laroussinie&Markey&Schnoebelen02,Kupferman&Vardi06}.
%

\textit{Our motivations.} In this paper, our main motivation is to analyze 
the effects of adding a binding mechanism with registers 
to specify runs of operational models such as pushdown systems and counter automata.
The registers are simple means to compare data values at different points of the execution.
Indeed, runs can be
naturally viewed as data words: for example, the finite alphabet is the set of control states and  the infinite alphabet
is the set of  data values (natural numbers, stacks, etc.). 
To do so, we enrich 
an ubiquitous logical formalism for model-checking
techniques, namely  linear-time temporal logic LTL, with registers. 
Even though this was the initial motivation
to introduce LTL with registers in~\cite{Demri&Lazic&Nowak06}, most decision problems considered 
in~\cite{Demri&Lazic&Nowak06,Lazic06,Demri&Lazic09} are essentially oriented towards satisfiability. 
In this paper, we focus on the following type of model-checking problem: given a set of
runs generated by an operational model, more precisely by a one-counter automaton, and a formula
from LTL with registers, is there a run satisfying the given formula? In our context, it will become
clear that the extension with two counters is
undecidable. It is not difficult to show
that this model-checking problem differs from those considered in~\cite{Lazic06,Demri&Lazic&Nowak06}
and from those in~\cite{Franceschet&deRijke&Schlingloff03,Franceschet&DeRijke06,tenCate&Franceschet05}
dealing with so-called hybrid logics.
%%
%% Ste 081008
%% and are of a different nature from those for hybrid logics investigated 
%% \iffossacs
%% in~\cite{Franceschet&DeRijke06,tenCate&Franceschet05}. 
%% \else
%% in~\cite{Franceschet&deRijke&Schlingloff03,Franceschet&DeRijke06,tenCate&Franceschet05}. 
%% \fi
%%
However, since  two consecutive counter values in a run are ruled by the
set of transitions, constraints on data that are helpful to get fine-tuned undecidability proofs for 
satisfiability problems in~\cite{Demri&Lazic&Nowak06,Demri&Lazic09} may not be allowed on runs.
This is precisely what we want to understand in this work. 
%% Ste 081008 (after Ranko's question)
%% Like in~\cite{Bouajjani&Jurski&Sighieanu07}, LTL with registers makes sense to specify
%% and reason about configurations of operational models, precisely counter systems. 
%% 
As a second main motivation, we would like to compare the results on LTL with registers
with those for first-order logic with data equality tests.
Indeed, LTL (with past-time operators) and first-order logic are equivalently expressive
by Kamp's theorem, but such a correspondence in presence of data values is not known. 
Our investigation about the complexity of model-checking one-counter automata with memoryful logics
include then first-order logic. 

\textit{Our contribution.} We study  complexity issues
related to the model-checking problem for LTL with registers
over one-counter automata that are  simple operational models, but
our undecidability results can be obviously
      lifted to pushdown systems when registers store the stack value.
Moreover, in order to determine borderlines for decidability, we 
also present results
for deterministic one-counter models that are less powerful 
but remain interesting when they are viewed as a mean to specify an infinite path 
on which model checking is performed, see analogous issues in~\cite{Markey&Schnoebelen03}.

We consider several
classes of one-counter
automata 
\iffossacs
(deterministic and nondeterministic)
\else
(deterministic, weakly deterministic and nondeterministic)
\fi 
and
several fragments by restricting the use of registers or the use of letters
from the finite alphabet. Moreover, we distinguish finite accepting runs
from infinite ones as data words. Unlike  results 
from~\cite{Ouaknine&Worrell06a,Ouaknine&Worrell07,Demri&Lazic09,Lazic06},
the decidability status of the model checking does not depend on the fact that we consider
finite data words instead of infinite ones. 
In this paper, we establish the following results.
\begin{itemize}
\itemsep 0 cm
\item Model checking LTL with registers [resp. \ first-order logic with data equality test]
       over deterministic one-counter automata  is \pspace-complete
      (see Sect.~\ref{section-pspace}). 
      \pspace-hardness is established 
      by reducing QBF and it also holds when no letters from the finite alphabet are used
      in formulae.
      \ifptime 
      When the number of registers [resp. variables] is bounded,
      the problem can be solved in polynomial time. 
      \fi 
      In order to get these complexity upper bounds,
      we  translate our problems into model-checking first-order logic without data equality test
      over ultimately periodic words that can be solved in polynomial space thanks 
      to~\cite{Markey&Schnoebelen03}. 
\item Model checking LTL with registers over nondeterministic one-counter automata restricted
      to a unique register and without alphabet is $\Sigma_1^1$-complete in the infinitary case
      by reducing the recurrence problem for Minsky machines (see Sect.~\ref{section-undecidability}).
      In  the finitary case, the problem is shown $\Sigma_1^0$-complete by reducing
      the halting problem for Minsky machines. 
      These results are quite surprising since several verification problems for one-counter
      automata are  decidable with relatively 
      low
      \iffossacs 
      complexity~\cite{Jancaretal04,Serre06}. 
       \else
      complexity~\cite{Jancaretal04,Serre06,Demri&Gascon07}. 
       \fi 
      Moreover, finitary satisfiability for LTL with one register  is decidable~\cite{Demri&Lazic09}
       even though with non-primitive recursive complexity. 
      %% 
      %% However, the temporal logic we consider herein is richer because of its ability to store
      %% a counter value. 
      %%
      These results can be also obtained for first-order logic with data equality test restricted to two variables
      by analysing the structure of formulae used in the undecidability proofs and by using~\cite{Demri&Lazic09}.
\end{itemize}
Figure~\ref{figure-table} contains a summary of the main results we obtained; notations are fully explained in 
Section~\ref{section-preliminaries}. For instance, 
$\fMC{\infin}{1}[\mynext,\sometimes]$ refers to the existential model-checking problem
on infinite accepting runs from one-counter automata with freeze LTL restricted to the temporal operators
``next'' and ``sometimes'', and to a unique register. Similarly, 
$\fotwoMC{\infin}{}[\sim,<]$ refers to 
the existential model-checking problem
on finite accepting runs from one-counter automata with first-order logic on data words restricted to two individual
variables. 

\begin{figure}
{\small 
 \begin{center}
\begin{tabular}{|c|c|c|}
\hline
\pspace-completeness  & $\Sigma_1^0$-completeness & $\Sigma_1^1$-completeness \\
 for det. 1CA & for 1CA & for 1CA \\
\hline
 $\fMC{\infin}{}$, $\fMC{\fin}{}$  &  $\fMC{\fin}{1}[\mynext,\sometimes]$ & 
      $\fMC{\infin}{1}[\mynext,\sometimes]$ \\
$\fMC{\infin}{}[\sometimes]$, $\fMC{\fin}{}[\mynext,\sometimes]$ & $\fpureMC{\fin}{1}[\mynext,\sometimes]$ & $\fpureMC{\infin}{1}[\mynext,\sometimes]$ \\ 
\hline
$\foMC{\infin}{}$, $\foMC{\fin}{}$ & $\fotwoMC{\fin}{}[\sim,<]$ & $\fotwoMC{\infin}{}[\sim,<]$ \\
$\foMC{\infin}{}[\sim,<]$ & & \\
\hline
\end{tabular}
\end{center}
} 
\caption{Summary of main results}
\label{figure-table}
\end{figure}

\iffossacs

Because of lack of space, omitted proofs can be found in \cite{long-version}.
\else
\textit{Plan of the paper.} In Sect.~\ref{section-preliminaries}, 
we introduce the model-checking problem for LTL with registers over
one-counter automata as well as the corresponding problem for
first-order logic with data equality test. 
In Sect.~\ref{section-decidability}, we consider decidability and complexity issues
for model checking deterministic one-counter automata. 
In Sect.~\ref{section-undecidability}, several model-checking problems over nondeterministic one-counter automata
are shown undecidable. \\

\noindent
This paper is an extended version of~\cite{Demri&Lazic&Sangnier08a}
that also improves significantly the results about the \pspace \ upper bounds and the undecidability results,
in particular by considering first-order language over data words.

\fi

 \section{Preliminaries}
\label{section-preliminaries}

%% Ste 140807
%% 
%% We denote by $\entiers$ the set of integers.
%% 
%
%%%%%%%%%%%%%%%%%%%%%%%%%%%%%%%%%%%%%%%%%%%%%%%%%%%%%%%%%%%%%%%%%%%%%%%%%%%%%
\subsection{One-counter automaton}
%%%%%%%%%%%%%%%%%%%%%%%%%%%%%%%%%%%%%%%%%%%%%%%%%%%%%%%%%%%%%%%%%%%%%%%%%%%%%
%

Let us recall standard definitions and notations about our operational models.
% \begin{definition}[one-counter automaton]
\iffossacs
A one-counter automaton is a tuple $ \aautomaton = \triple{\locs,\aloc_I}{\delta}{F}$ where
$\locs$ is a finite set of states,
$\aloc_{I} \in \locs$ is the initial state,
$F \subseteq \locs$ is the set of accepting states and
$\delta \subseteq \locs \times L \times \locs$ is the transition relation
over the instruction set $L =  \set{\mathtt{inc, dec, ifzero}}$.
\else
A one-counter automaton is a tuple $ \aautomaton = \triple{\locs,\aloc_I}{\delta}{F}$ where:
\begin{itemize}
\itemsep 0 cm 
\item $\locs$ is a finite set of states,
\item $\aloc_{I} \in \locs$ is the initial state,
\item $F \subseteq \locs$ is the set of accepting states,
\item $\delta \subseteq \locs \times L \times \locs$ is the transition relation
over the instruction set $L =  \set{\mathtt{inc, dec, ifzero}}$.
\end{itemize} 
\fi
% \end{definition}
%
%
A counter valuation $\aregval$ is  an element of $\Nat$ 
and a configuration of $\aautomaton$ is a pair in $\locs \times \Nat$. The initial configuration is the pair 
$\pair{q_{I}}{0}$. 
As usual, a one-counter automaton $\aautomaton$ induces a (possibly infinite) transition
system $\pair{\locs \times \Nat}{\step{}}$ such that $\pair{\aloc}{n} \step{} \pair{\aloc'}{n'}$
iff one of the conditions below holds true:
\iffossacs
(1) $\triple{\aloc}{\mathtt{inc}}{\aloc'} \in \delta$ and $n' = n+1$,
(2) $\triple{\aloc}{\mathtt{dec}}{\aloc'} \in \delta$ and $n' = n-1$ (and $n' \in \Nat$),
(3) $\triple{\aloc}{\mathtt{ifzero}}{\aloc'} \in \delta$ and $n = n' = 0$. 
\else
\begin{enumerate}
\itemsep 0 cm
\item $\triple{\aloc}{\mathtt{inc}}{\aloc'} \in \delta$ and $n' = n+1$,
\item $\triple{\aloc}{\mathtt{dec}}{\aloc'} \in \delta$ and $n' = n-1$ (and $n' \in \Nat$),
\item $\triple{\aloc}{\mathtt{ifzero}}{\aloc'} \in \delta$ and $n = n' = 0$. 
\end{enumerate}
\fi
%%Arn 101007
%%We write $\step{*}$ to denote the reflexive and transitive closure of $\step{}$. 
A finite [resp.\ infinite] \emph{run} $\arun$ 
is a finite [resp.\ infinite] sequence
$\arun = \pair{\aloc_0}{n_0} \step{} \pair{\aloc_1}{n_1} \step{} \cdots$
where $\pair{\aloc_0}{n_0}$ is the initial configuration. 
%% Arn170308
%%A finite run is \emph{accepting} iff it ends with an accepting location.
A finite run $\arun=\pair{\aloc_0}{n_0} \step{} \pair{\aloc_1}{n_1} \step{} \cdots \step{} \pair{\aloc_f}{n_f}$ is \emph{accepting} iff $\aloc_f$ is an accepting state. 
%% and for all $0 \leq i < f$, $\aloc_i$ is not an accepting location. 
An infinite run $\arun$ is accepting iff it contains an accepting state
infinitely often (B\"uchi acceptance condition).
%% Arn 081007
%%Similarly, a \emph{path} is defined as a run except that we do not require that the first
%%configuration is initial: a path is simply a finite path in the induced transition system.
\iffossacs \else All these notations can be naturally adapted to multicounter automata. \fi
%% Ste 140807
%% The behavior of a one-counter machine $M$ is given by the mean of a relation 
%% $\rightarrow_{M} \subseteq (Q \times \entiers) \times (Q \times \entiers)$ defined by~:
%% \begin{itemize} 
%%  \item $(q,n) \rightarrow_{M} (q',n+1) $ if and only if there exists $(q,inc,q') \in \delta$,
%% \item $(q,n) \rightarrow_{M} (q',n-1) $ if and only if there exists $(q,dec,q') \in \delta$ and $n>0$,
%% \item $(q,0) \rightarrow_{M} (q',0) $ if and only if there exists $(q,zero,q') \in \delta$.
%% \end{itemize}
%% We denote by $\rightarrow_{M}^\ast$ the reflexive and transitive closure of $\rightarrow_{M}$.
%% Ste 140807
%% %
%% 
%% 
%% We say that a one-counter automaton $\aautomaton$ 
%% is \emph{deterministic} if and only if $F=Q$ and for all configuration $(q,n) \in Q \times \entiers$, there 
%% exists an unique configuration $(q',n') \in Q \times \entiers$ such that $(q,n) \rightarrow_{M} (q',n')$. This 
%% corresponds roughly to deterministic one-counter Minksy Machine.
%% %
%% And $M=\langle Q, \delta, q_0, F \rangle$ is \emph{weakly deterministic} if and only if for all $(q,inst,q'), 
%% \linebreak[0] (q,inst, q") \in \delta$, $q' = q"$. 
%% %

A one-counter automaton $\aautomaton$ is \emph{deterministic} whenever it 
corresponds to a deterministic one-counter Minsky machine: for every state $\aloc$, 
\iffossacs
either $\aautomaton$ has a unique transition from $\aloc$ incrementing the counter, 
or $\aautomaton$ has exactly two transitions from $\aloc$, one with
instruction $\mathtt{ifzero}$ and the other with instruction
$\mathtt{dec}$,
or $\aautomaton$ has no transition from $\aloc$ (not present in original deterministic
Minsky machines). 
\else
\begin{itemize}
\itemsep 0 cm
\item either $\aautomaton$ has a unique transition from $\aloc$ incrementing the counter, 
\item or $\aautomaton$ has exactly two transitions from $\aloc$, one with
      instruction $\mathtt{ifzero}$ and the other with instruction
      $\mathtt{dec}$,
\item or $\aautomaton$ has no transition from $\aloc$ (not present in original deterministic
Minsky machines~\cite{Minsky67}). 
\end{itemize}
\fi
In the transition system induced by any deterministic one-counter automaton, 
each configuration has at most one successor.
One-counter automata in full generality are understood as 
\emph{nondeterministic} one-counter automata.

% A one-counter automaton $\aautomaton$ is \emph{weakly deterministic} whenever 
% for every location $\aloc$, if $\triple{\aloc}{l}{\aloc'}, \triple{\aloc}{l'}{\aloc''}
% \in \delta$, we have $l = l'$ implies $\aloc' = \aloc''$. 
% Observe that the transition systems induced by weakly deterministic one-counter automata
% are not necessarily deterministic.
% Finally, one-counter automata in full generality are understood as 
% \emph{nondeterministic} one-counter automata. 

%%%%%%%%%%%%%%%%%%%%%%%%%%%%%%%%%%%%%%%%%%%%%%%%%%%%%%%%%%%%%%%%%%%%%%%%%%%%%
\subsection{LTL over data words}
%%%%%%%%%%%%%%%%%%%%%%%%%%%%%%%%%%%%%%%%%%%%%%%%%%%%%%%%%%%%%%%%%%%%%%%%%%%%%

%% 
%% %
%% \commAS{I have taken almost the same text as in your paper of LICS 06}\\
%% %
%% 

Formulae of the logic $\fLTL{\downarrow,\aalphabet}{}$~\cite{Demri&Lazic09}
where $\aalphabet$ is a finite alphabet are defined as follows:
$$
\begin{array}{lcl}
\aformula & ::= &\aletter~ \mid~ \uparrow_r \ \mid \ \neg \aformula~ \mid~ 
\aformula \wedge \aformula~ \mid~ \aformula \until \aformula~ \mid~ \mynext  \aformula~ \mid~ 
\downarrow_r \aformula
\end{array}
$$
where $\aletter \in \aalphabet$ and $r$ ranges over  $\POSentiers$.
We write $\fLTL{\downarrow}{}$ to denote LTL with registers for some unspecified finite alphabet.
An occurrence of $\uparrow_r $ within the scope of some freeze quantifier $\downarrow_r$ is bound by it; 
otherwise it is free. A sentence is a formula with no free occurrence of any $\uparrow_r $. 
%% Ste 021007
% Given a set $\mathcal{O}$ 
% of temporal operators definable from $\until$, $\mynext$, $\wedge$ and $\neg$, and a natural number
% $n > 0$,
% we write  $\fLTL{\downarrow,\aalphabet}{n}(\mathcal{O})$ to denote the restriction
% of  $\fLTL{\downarrow,\aalphabet}{}$ to temporal operators in $\mathcal{O}$ and
% to registers in $\set{1,\ldots,n}$. 
Given  a natural number
$n > 0$,
we write  $\fLTL{\downarrow,\aalphabet}{n}$ to denote the restriction
of  $\fLTL{\downarrow,\aalphabet}{}$ to registers in $\set{1,\ldots,n}$. 
%% 
%% Ste 150807
%% those in $\{\textsf{U},\textsf{X}\}$ and $n > 0$, we denote by \boundfreezeLTLdefwithO{n}{\mathcal{O}} the
%%  fragment of \freezeLTLdef with temporal operators in $\mathcal{O}$ and with $n$ registers, i.e. where 
%% $r\in \{1,...,n\}$.
%% 
%% 
Models of $\fLTL{\downarrow,\aalphabet}{}$  are \emph{data words}. 
A data word $\adataword$ 
over a finite alphabet $\aalphabet$ is a non-empty word in $\finwords$ or $\infwords$,
 together with an equivalence relation $\sim^{\adataword}$  on word indices. 
We write $\length{\adataword}$ for the 
length of the data word, $\adataword(i)$ for its letters where $0 \leq i < \length{\adataword}$.
\iffossacs
\else
Let $\finwords(\sim)$ 
[resp.\ $\infwords(\sim)$] denote the sets of all such finite [resp.\ infinite] data words. We denote by $\finandinfwords(\sim)$ the set $\finwords(\sim) \cup \infwords(\sim)$ of finite and infinite data words.
\fi
%
%% Arn 081007
%%For a data word $\adataword$, let $\tt{str}(\adataword)$ be the underlying word in $\aalphabet^{\leq \omega}$.

A \emph{register valuation} $\aregval$ for a data word 
$\adataword$ is a finite partial map from $\POSentiers$ to the indices of $\adataword$. 
Whenever $\aregval(r)$ is undefined, the formula $\uparrow_r$ is interpreted as false. 
%% Ste 150807
%% An undefined register value in an atomic formula will make the latter false. 
%% Undefined register values will be used for initial automata states. 
Let $\adataword$ be a data word in $\finandinfwords(\sim)$  and $0 \leq i < \length{\adataword}$, the satisfaction relation $\models$ is defined as follows (Boolean clauses are omitted). 
$$
\begin{array}{rcl}
\adataword,i \models_{\aregval}~ \aletter & ~\defeq~ & \adataword(i)= \aletter \\
\adataword,i \models_{\aregval}~ \uparrow_r  & ~\defeq~ & r \in \domain{\aregval} 
\mbox{ and } \aregval(r) \sim^{\adataword} i \\
\adataword,i \models_{\aregval}~ \mynext \aformula & ~\defeq~ & i+1 < \length{\adataword}  \mbox{ and } 
\adataword,i+1 \models_{\aregval}~ \aformula \\
\adataword,i \models_{\aregval}~ \aformula_1 \until \aformula_2 & ~\defeq~ & 
\mbox{ for some }  i \leq j < \length{\adataword}, ~ \adataword,j \models_{\aregval}~ \aformula_2 \\
& & \mbox{ and for all } i \leq j' <j, \mbox{ we have } ~ \adataword,j' \models_{\aregval}~ \aformula_1\\
\adataword,i \models_{\aregval}~ \downarrow_r \aformula & ~\defeq~ & \adataword,i 
\models_{\aregval[r \mapsto i]}~ \aformula \\
\end{array}
$$
$\aregval[r \mapsto i]$ denotes the register valuation equal to $\aregval$ 
except that the register $r$ is mapped to the position $i$. 
In the sequel, we omit the subscript ``$\aregval$'' in 
$\models_{\aregval}$ when sentences are involved. 
We use the standard abbreviations for the temporal operators 
($\always$, $\sometimes$, $\always^{+}$, $\sometimes^{+}$, 
\ldots) and for the Boolean operators and constants ($\vee$, $\Rightarrow$, 
$\top$, $\perp$, \ldots).
The finitary [resp.\ infinitary] 
satisfiability problem for LTL with registers, noted $\fin$-SAT-LTL$^{\downarrow}$ [resp.\ 
 $\infin$-SAT-LTL$^{\downarrow}$],
is defined as
follows: 
\iffossacs
given  a finite alphabet $\aalphabet$ and a formula $\aformula$ in 
              $\fLTL{\downarrow,\aalphabet}{}$, is there an infinite [resp.\ a finite]
                 data word $\adataword$ such that $\adataword, 0 \models \aformula$?
\else
\begin{description}
\itemsep 0 cm
\item[Input:] A finite alphabet $\aalphabet$ and a formula $\aformula$ in 
              $\fLTL{\downarrow,\aalphabet}{}$;
\item[Question:] Is there a finite [resp.\ an infinite]
                 data word $\adataword$ such that $\adataword, 0 \models \aformula$?
\end{description}
\fi

\begin{theorem} \cite[Theorem 5.2]{Demri&Lazic09} \label{th:undecid-ltl}
$\fin$-SAT-LTL$^{\downarrow}$ restricted to  one register 
is decidable with non-primitive recursive complexity and $\infin$-SAT-LTL$^{\downarrow}$ restricted to  one register is $\Pi^0_1$-complete.
%% Ste 150807
%% The satisfiability of \boundfreezeLTLdef{1} over infinite data words is highly undecidable.
\end{theorem}
%
%% Ste 021007
%% By contrast, the finitary satisfiability problem for $\fLTL{\downarrow,\aalphabet}{1}$
%% is decidable~\cite{Demri&Lazic06}.
%%
 
Given a one-counter automaton $\aautomaton = \triple{\locs,\aloc_I}{\delta}{F}$,
finite [resp.\ infinite] accepting runs of $\aautomaton$ can be viewed as finite [resp.\ infinite]
data words over the alphabet $\locs$. Indeed,  given a run $\arun$, the equivalence relation
$\sim^{\arun}$ is defined as follows: $i \sim^{\arun} j$ iff the counter value at the $i$th position
of $\arun$ is equal to the counter value at the $j$th position of $\arun$.
In order to ease the presentation, in the sequel we sometimes store counter values in registers, which 
is an equivalent way to proceed by slightly adapting the semantics for $\uparrow_{\aregister}$ and 
 $\downarrow_{\aregister}$, and the values stored in registers (data). 

The finitary [resp.\ infinitary] (existential) model-checking
problem over one-counter automata for LTL with registers, noted $\fMC{\fin}{}$ [resp.\
$\fMC{\infin}{}$], 
is defined as
follows:
\iffossacs 
given a one-counter automaton $\aautomaton = \triple{\locs,\aloc_I}{\delta}{F}$
              and  a sentence $\aformula$ in 
              $\fLTL{\downarrow,Q}{}$, 
is there a finite [resp.\ infinite] accepting run $\arun$ of $\aautomaton$
                 such that $\arun, 0 \models \aformula$?
                 If the answer is ``yes'', we write $\aautomaton \models^{< \omega} \aformula$
                 [resp.\ $\aautomaton \models^{\omega} \aformula$]. 
\else
\begin{description}
\itemsep 0 cm
\item[Input:] A one-counter automaton $\aautomaton = \triple{\locs,\aloc_I}{\delta}{F}$
              and a  sentence $\aformula$ in 
              $\fLTL{\downarrow,Q}{}$;
\item[Question:] Is there a finite [resp.\ infinite] accepting run $\arun$ of $\aautomaton$
                 such that $\arun, 0 \models \aformula$?
                 If the answer is ``yes'', we write $\aautomaton \models^{\fin} \aformula$
                 [resp.\ $\aautomaton \models^{\infin} \aformula$]. 
\end{description}
\fi
%% Ste 290907
%% Since $\aformula$ is a sentence, without any loss of generality, we can assume that
%% the initial register valuation maps every register to zero value.
In this existential version of model checking, this problem can be
 viewed as  a variant of satisfiability in which satisfaction of a formula
can be only witnessed within a specific class of data words, namely the
accepting runs of the automata. Results
for the universal version of model checking will follow easily from those for the existential version.

We write $\fMC{\alpha}{n}$ to denote the restriction
of $\fMC{\alpha}{}$ to formulae with at most $n$ registers.
Very often, it makes sense that only counter values are known but not the current state
of a configuration, which can be understood as an internal information about the system.
We write  $\fpureMC{\alpha}{n}$ to denote the restriction
of  $\fMC{\alpha}{n}$  (its ``pure data'' version)
to formulae with atomic formulae  only of the form $\uparrow_r$. Given a set $\mathcal{O}$  of temporal operators, we write $\fMC{\alpha}{n}[\mathcal{O}]$ [resp. $\fpureMC{\alpha}{n}[\mathcal{O}]$] to denote
the restriction of $\fMC{\alpha}{n}$ [resp. $\fpureMC{\alpha}{n}$] to formulae using only temporal operators in $\mathcal{O}$.

\begin{example} Here are some properties that can be stated in
 $\fLTL{\downarrow,Q}{2}$ along a run.
\begin{itemize}
\itemsep 0 cm
\item ``There is a suffix such that all the counter values are different'':
      $$\sometimes \always (\downarrow_1 \always^{+} \neg \uparrow_1).$$
%% \item ``There is a suffix such that the sequence of instructions is either
%%        $(\mathtt{inc} \cdot \mathtt{dec})^{\omega}$ or  $(\mathtt{dec} \cdot \mathtt{inc})^{\omega}$'': 
%%       $\sometimes \always (\downarrow_1 \mynext (\mynext \uparrow_1 \wedge \neg \uparrow_1))$. 
\item ``Whenever state $\aloc$ is reached with current counter value $n$ and next current counter value $m$,
        if there is a next occurrence of $\aloc$, the two consecutive counter values are also $n$ and $m$'':
      $$\always(\aloc \Rightarrow \downarrow_1 \mynext \downarrow_2 \mynext \always 
       (\aloc \Rightarrow \uparrow_1 \wedge~
       \mynext \uparrow_2)).$$
\end{itemize} 
\end{example}

%%
%% Ste 081007
%%
%% In the sequel, we shall investigate the decidability/complexity status of the
%% model-checking problems over deterministic and nondeterministic
%% one-counter automata. For most results 
%% about undecidability  and complexity lower bounds, we
%% shall obtain them for the pure version of the model-checking problems. 
%% 

\iffossacs
\else
Observe also that  we have chosen as alphabet the set of states of the automata.
Alternatively, it would have been possible to add  finite alphabets to automata, to label each transition
by a letter and then consider as data words generated from automata the recognized
words augmented with the counter values. This choice does not change our main results  but it improves the readability of some technical details.
%This does not make any essential difference
%with our current choice  that simplifies a little some technical developements. 
\fi

\subsection{First-order logic over data words}

Let us introduce the second logical formalism considered in the paper.
Formulae of $\fo{\aalphabet}{\sim,<,+1}$~\cite{Bojanczyketal06a} 
where $\aalphabet$ is a finite alphabet are defined as follows:
$$
\aformula ::= \aletter(\avariable) \ \mid \
              \avariable \sim \avariablebis \ \mid \
              \avariable < \avariablebis \ \mid \
              \avariable = \avariablebis + 1 \ \mid \
              \neg \aformula \ \mid \
              \aformula \wedge \aformula \ \mid \
              \exists \ \avariable \ \aformula
$$
where $\aletter \in \aalphabet$ and $\avariable$ ranges over a countably infinite set of variables.
We write $\fo{}{\sim,<,+1}$ to denote $\fo{\aalphabet}{\sim,<,+1}$
for some unspecified finite alphabet and $\fo{}{<,+1}$ to denote the restriction of
$\fo{}{\sim,<,+1}$ without atomic formulae of the form $\avariable \sim \avariablebis$. 
Given  a natural number
$n > 0$,
we write  $\fon{\aalphabet}{n}(\sim,<,+1)$ to denote the restriction
of   $\fo{\aalphabet}{\sim,<,+1}$ to variables in $\set{\avariable_1,\ldots,\avariable_n}$. 
A variable valuation $\avarval$ for a data word $\adataword$ is a finite partial map
from the set of variables
to the indices of $\adataword$. Let $\adataword$ be a data word in $\finandinfwords(\sim)$, the satisfaction relation $\models$ is defined as follows (Boolean clauses
are again omitted):
$$
\begin{array}{rcl}
\adataword \models_{\avarval}~ \aletter(\avariable) & ~\defeq~ & 
\avarval(\avariable) \ \mbox{is  defined  and} \ 
\adataword(\avarval(\avariable))= \aletter \\
\adataword \models_{\avarval}~ \avariable \sim \avariablebis   & ~\defeq~ & 
\avarval(\avariable) \ \mbox{and} \  \avarval(\avariablebis) \ \mbox{are defined and} \ 
 \avarval(\avariable) \sim^{\adataword}  \avarval(\avariablebis) \\
\adataword \models_{\avarval}~ \avariable < \avariablebis   & ~\defeq~ & 
\avarval(\avariable) \ \mbox{and} \  \avarval(\avariablebis) \ \mbox{are defined and} \ 
 \avarval(\avariable) < \avarval(\avariablebis) \\
\adataword \models_{\avarval}~ \avariable = \avariablebis + 1    & ~\defeq~ & 
\avarval(\avariable) \ \mbox{and} \  \avarval(\avariablebis) \ \mbox{are defined and} \ 
 \avarval(\avariable) = \avarval(\avariablebis) + 1\\
\adataword \models_{\avarval}~ \exists \ \avariable \ \aformula   & ~\defeq~ & 
\mbox{there is} \ i \in \Nat \ \mbox{such that} \  0 \leq i < \length{\adataword} \ \mbox{and} \  
\adataword \models_{\avarval[\avariable \mapsto i]}~\aformula\\
\end{array}
$$
$\avarval[\avariable \mapsto i]$ denotes the variable valuation equal to $\avarval$ 
except that the variable $\avariable$ is mapped to the position $i$. 
In the sequel, we omit the subscript ``$\avarval$'' in 
$\models_{\avarval}$ when sentences are involved.

The finitary [resp.\ infinitary] (existential) model-checking
problem over one-counter automata for the logic $\fo{\aalphabet}{\sim,<,+1}$, noted $\foMC{\fin}{}$ [resp.\
$\foMC{\infin}{}$]   
is defined as
follows:
\begin{description}
\itemsep 0 cm
\item[Input:] A one-counter automaton $\aautomaton$
              and  a sentence $\aformula$ in 
              $\fo{Q}{\sim,<,+1}$;
\item[Question:] Is there a finite [resp.\ infinite] accepting run $\arun$ of $\aautomaton$
                 such that $\arun \models \aformula$?
                 If the answer is ``yes'', we write $\aautomaton \models^{\fin} \aformula$
                 [resp.\ $\aautomaton \models^{\infin} \aformula$]. 
\end{description}

We write $\foMC{\alpha}{n}$ to denote the restriction
of $\foMC{\alpha}{}$ to formulae with at most $n$ variables.
We write  $\fopureMC{\alpha}{n}$ to denote the restriction
of  $\foMC{\alpha}{n}$  (its ``pure data'' version)
to formulae with no atomic formulae  of the form $\aletter(\avariable)$.

Extending the standard translation from LTL into first-order logic, we can easily establish the result
below.

\begin{lemma} \label{lemma-ltl-fo}
Given a sentence $\aformula$ in $\fLTL{\downarrow,\aalphabet}{n}$, there is
a first-order formula $\aformula'$ in $\fo{\aalphabet}{\sim,<,+1}$ that can be computed in
linear time in $\length{\aformula}$ such that
\begin{enumerate}
\itemsep 0 cm
\item $\aformula'$ has at most $max(3,n+1)$ variables,
\item $\aformula'$ has a unique free variable, say $\avariablebis_0$,
\item for all data words $\adataword$, register valuations $\aregval$ and $i \geq 0$,
we have $\adataword, i \models_{\aregval} \aformula$ iff
$\adataword \models_{\avarval} \aformula'$, where 
for $r \in \set{1,\ldots,n}$, $\aregval(r) = \avarval(\avariable_r)$
and $\avarval(\avariablebis_0) = i$. 
\end{enumerate}
\end{lemma}

\begin{proof}
We build a translation function $T$ which takes as arguments a formula in $\fLTL{\downarrow,\aalphabet}{n}$ and a variable, and which 
returns the wanted formula in $\fo{\aalphabet}{\sim,<,+1}$. Intuitively the variable, which is given as argument, is used to represent 
the current position in the data word. Then, we  use the variables $\avariable_1,\ldots,\avariable_r$ to characterize the registers. 
 We add to this set of variables three variables $\avariablebis_0,\avariablebis_1$ and $\avariablebis_2$. In the sequel, we  write
 $\avariablebis$ to represent indifferently $\avariablebis_0$ or $\avariablebis_1$ or $\avariablebis_2$. Furthermore the notation 
$\avariablebis_{i+1}$ stands for $\avariablebis_{(i+1)mod(3)}$ and $\avariablebis_{i+2}$ stands for $\avariablebis_{(i+2)mod(3)}$. 
The function $T$, which is homomorphic for the Boolean operators, is defined inductively as follows, for $i \in \set{0,1,2}$:
\begin{itemize}
\item $T(\aletter,\avariablebis)=\aletter(\avariablebis)$,
\item $T(\uparrow_r,\avariablebis)=\avariablebis \sim \avariable_r$,
\item $T(\mynext \aformula,\avariablebis_i)=\exists \ \avariablebis_{i+1} \ (\avariablebis_{i+1}=\avariablebis_i + 1 \wedge 
T(\aformula,\avariablebis_{i+1}))$,
\item $T(\aformula \until \aformulabis,\avariablebis_i)=\exists \ \avariablebis_{i+1} \ (\avariablebis_{i} \leq \avariablebis_{i+1}
 \wedge T(\aformulabis,\avariablebis_{i+1}) \wedge \forall \ \avariablebis_{i+2} \ (\avariablebis_{i} \leq \avariablebis_{i+2} < 
\avariablebis_{i+1} \Rightarrow T(\aformula,\avariablebis_{i+2}))$,
\item $T(\downarrow_r \aformula,\avariablebis)=\exists \ \avariable_r \ (\avariable_r=\avariablebis \wedge T(\aformula,\avariablebis))$. 
\end{itemize}
Then if $\aformula$ is a formula in $\fLTL{\downarrow,\aalphabet}{n}$ and $\avariablebis_0$ is the variable chosen to characterize 
the current position in the word, the formula $T(\aformula,\avariablebis_0)$ satisfies the three conditions given in the above lemma. 
In order to ensure the first condition, we use the fact that we can recycle the variables. More details about this technique can be 
found in \cite{gabay-expressive-81}.
\hfill $\Box$
%%\qed
\end{proof}

The decidability borderline for $\fo{}{\sim,<,+1}$ is between two and three variables.

\begin{theorem} \cite[Theorem 1, Propositions 19 \& 20]{Bojanczyketal06a} \label{th:undecid-fo}
Satisfiability for $\fo{}{\sim,<,+1}$ restricted to 3 variables is undecidable 
and satisfiability for $\fon{}{2}(\sim,<,+1)$ is decidable (for both finitary and infinitary cases).
\end{theorem}

In Section~\ref{section-decidability} we will use Theorem~\ref{theorem-without-data} below in an 
essential way. 

\begin{theorem} \label{theorem-without-data} \cite[Proposition 4.2]{Markey&Schnoebelen03}
Given two finite words $s,t \in \aalphabet^{*}$ and
a sentence $\aformula$ in $\fo{\aalphabet}{<,+1}$, checking whether
$s \cdot t^{\omega} \models \aformula$ can be done in space $\mathcal{O}((\length{s} + \length{t}) \times
\length{\aformula}^2)$.
\end{theorem}

\ifptime
Moreover, for every fixed number of variables $n \geq 0$, checking whether 
$s \cdot t^{\omega} \models \aformula$ where $\aformula$ has at most $n$ variables, can be done
in polynomial time. 
\fi

\iffossacs
\else
%%%%%%%%%%%%%%%%%%%%%%%%%%%%%%%%%%%%%%%%%%%%%%%%%%%%%%%%%%%%%%%%%%%%%%%%%%%%%%
\subsection{Purification of the model-checking problem}
\label{section-purification}
%%%%%%%%%%%%%%%%%%%%%%%%%%%%%%%%%%%%%%%%%%%%%%%%%%%%%%%%%%%%%%%%%%%%%%%%%%%%%%
%
\fi 

We now show how to get rid of propositional variables
by reducing the model-checking problem over one-counter automata to its pure version. 
This amounts to transform any  $\fMC{}{}$
  instance  into a $\fpureMC{}{}$ instance.

\begin{lemma}[{\small Purification for $\fLTL{\downarrow}{}$}]
\label{lemma-purification}
Given a one-counter automaton $\aautomaton$
              and  a sentence $\aformula$ in 
              $\fLTL{\downarrow,\locs}{n}$, one can compute in logarithmic space in
$\length{\aautomaton} + \length{\aformula}$ a one-counter automaton 
$\aautomaton_P$ and a formula $\aformula_P$ in $\fLTL{\downarrow,\emptyset}{max(n,1)}$   
such that  $\aautomaton \models^{\fin} \aformula$
                 [resp.\ $\aautomaton \models^{\infin} \aformula$] iff 
$\aautomaton_P \models^{\fin} \aformula_P$
                 [resp.\ $\aautomaton_P \models^{\infin} \aformula_P$].
Moreover, $\aautomaton$ is deterministic 
% [resp.\ weakly deterministic] 
iff $\aautomaton_P$  is deterministic.
% [resp.\ weakly deterministic]. 
\end{lemma}

The idea of the proof is simply to identify states with patterns about the changes
of the unique counter that can be expressed in 
$\fLTL{\downarrow,\emptyset}{}$.
\begin{proof}
% \begin{proof}(sketch)
Let $\aautomaton=\triple{\locs,\aloc_I}{\delta}{F}$ with $\locs=\set{\aloc_1,\ldots,\aloc_m}$ and
$\aformula$ be an $\fLTL{\downarrow,Q}{}$ formula. In order to define $\aautomaton_P$, 
we identify states with patterns about the changes of the unique counter.
Let $\aautomaton_P$ be $\triple{\locs_P,\aloc_I}{\delta_P}{F_P}$ with $\locs_P= \locs \uplus \locs'$ and $\locs'$
is defined below: 
$$
\begin{array}{ll}
\locs'=&\bigset{\aloc_i^1,\aloc_i^2,\aloc_i^3,\aloc_i^4,\aloc_i^5,\aloc_{i,F} \mid i \in \set{1,\ldots,m}}\\
& \cup \bigset{\aloc_{i,j},\aloc'_{i,j} \mid i \in \set{1,\ldots,m} \mbox{ and } j \in  \set{1,\ldots,m+1} \mbox{ and } i \neq j}\\
& \cup \bigset{\aloc_{i,i}^0,\aloc_{i,i},\aloc_{i,i}^1,\aloc_{i,i}^2 \mid i \in \interval{1}{m}}.
\end{array}
$$
Figure~\ref{fig:purif1} presents the set of transitions $\delta_P$ associated with each state $\aloc_i$  
of $\locs$ (providing a pattern). Furthermore, for all $i,j \in \interval{1}{m}$, 
$\aloc_{i,F} \step{\mathtt{a}} \aloc_j \in \delta_P$
iff  $\aloc_i \step{\mathtt{a}} \aloc_j \in \delta$.
The sequence of transitions  associated to each $\aloc_i \in \locs$ is a sequence of $m+2$ picks and among these
picks, the first pick is the only one of height $3$, the $i$-th pick is the only one of height $2$, and 
the height of all the other 
picks is $1$. Observe that this sequence of transitions has a fixed length and it is composed of exactly $9+2(m+1)$ states.

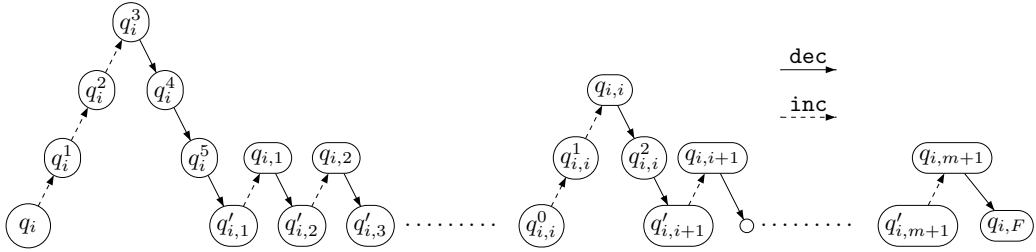
\begin{figure}[htbp]
\begin{center}
\scalebox{0.9}{
\begin{picture}(145,35)(0,0)
%
%\unitlength=1mm
%
\node[Nw=6.5,Nh=6.5](q)(0,0){$\aloc_{i}$}
\node[Nadjust=wh](qi1)(5,10){$\aloc_i^1$}
\drawedge[dash={0.7}0](q,qi1){}
\node[Nadjust=wh](qi2)(10,20){$\aloc_i^2$}
\drawedge[dash={0.7}0](qi1,qi2){}
\node[Nadjust=wh](qi3)(15,30){$\aloc_i^3$}
\drawedge[dash={0.7}0](qi2,qi3){}
\node[Nadjust=wh](qi4)(20,20){$\aloc_i^4$}
\drawedge(qi3,qi4){}
\node[Nadjust=wh](qi5)(25,10){$\aloc_i^5$}
\drawedge(qi4,qi5){}
\node[Nadjust=wh](qi10)(30,0){$\aloc'_{i,1}$}
\drawedge(qi5,qi10){}
\node[Nadjust=wh](qi11)(35,10){$\aloc_{i,1}$}
\drawedge[dash={0.7}0](qi10,qi11){}
\node[Nadjust=wh](qi20)(40,0){$\aloc'_{i,2}$}
\drawedge(qi11,qi20){}
\node[Nadjust=wh](qi21)(45,10){$\aloc_{i,2}$}
\drawedge[dash={0.7}0](qi20,qi21){}
\node[Nadjust=wh](qi30)(50,0){$\aloc'_{i,3}$}
\drawedge(qi21,qi30){}
\put(55,0){$\ldots \ldots \ldots$}
\node[Nadjust=wh](qii0)(75,0){$\aloc^0_{i,i}$}
\node[Nadjust=wh](qii1)(80,10){$\aloc^1_{i,i}$}
\drawedge[dash={0.7}0](qii0,qii1){}
\node[Nadjust=wh](qii)(85,20){$\aloc_{i,i}$}
\drawedge[dash={0.7}0](qii1,qii){}
\node[Nadjust=wh](qii2)(90,10){$\aloc^2_{i,i}$}
\drawedge(qii,qii2){}
\node[Nadjust=wh](qii3)(95,0){$\aloc'_{i,i+1}$}
\drawedge(qii2,qii3){}
\node[Nadjust=wh](qii4)(100,10){$\aloc_{i,i+1}$}
\drawedge[dash={0.7}0](qii3,qii4){}
\node[Nadjust=wh](qii5)(105,0){}
\drawedge(qii4,qii5){}
\put(107,0){$\ldots \ldots \ldots$}
\node[Nadjust=wh](qin0)(130,0){$\aloc'_{i,m+1}$}
\node[Nadjust=wh](qin1)(135,10){$\aloc_{i,m+1}$}
\drawedge[dash={0.7}0](qin0,qin1){}
\node[Nadjust=wh](qif)(143,0){$\aloc_{i,F}$}
\drawedge(qin1,qif){}

\node[Nw=0,Nh=0](inc1)(110,16){}
\node[Nw=0,Nh=0](inc2)(118,16){}
\drawedge[dash={0.7}0](inc1,inc2){$\inc$}

\node[Nw=0,Nh=0](dec1)(110,23){}
\node[Nw=0,Nh=0](dec2)(118,23){}
\drawedge(dec1,dec2){$\dec$}

\end{picture}
}
\end{center}
\caption{Encoding $\aloc_i$ by a pattern made of $m+2$  picks and of length $9+2(m+1)$}
\label{fig:purif1}
\end{figure}

Finally, the set of accepting states of $\aautomaton_P$ is defined as the set
$\set{\aloc_{i,F} \mid \aloc_i \in F}$.
%Arn 261008
% Finally, the set of accepting locations of $\aautomaton_P$ is defined as the set
% $\set{\aloc'_{i,1}: \aloc_i \in F, \ i \neq 1}$ augmented with $\aloc^{0}_{1,1}$ if $\aloc_1 \in F$. 
% An apparently obvious definition would be to consider as accepting locations those
% in the set $\set{\aloc_{i,F}: \aloc_i \in F}$ but we need to add 6 more steps after reaching $\aloc_{i,F}$ 
% to identify this location with equality tests on data values only. 
%
%% Ste 081008
%% $F_P=\set{\aloc_{i,F} \mid i \in \set{1,\ldots,m} 
%% \mbox{ and } \aloc_i \in F}$. 
%%
In order to detect the first pick of height $3$ which characterizes the beginning of the sequence of 
transitions associated to each state belonging to $\locs$, we build the two following formulae 
in $\fLTL{\downarrow,\emptyset}{1}$:
\begin{itemize}
\itemsep 0 cm
\item $\aformulater_{\neg 3/7}$ which expresses that ``among the 7 next counter values (including the current counter value), there are no 3 equal values'',
\item $\aformulater_{0 \sim 6}$ which expresses that ``the current counter value is equal to the counter value at the 6th next position''.
\end{itemize}
These two formulae can be written as follows:
$$
\begin{array}{lcl}
\aformulater_{\neg 3/7} & = & \neg \big( \downarrow_1 \big(\bigvee_{i \neq j \in \interval{1}{6}} (\mynext^i \uparrow_1 \wedge \mynext^j \uparrow_1) \big) \\
& & \vee \mynext \downarrow_1 \big(\bigvee_{i \neq j \in \interval{1}{5}} (\mynext^i \uparrow_1 \wedge \mynext^j \uparrow_1) \big) \\
& &  \vee \mynext^2 \downarrow_1 \big(\bigvee_{i \neq j \in \set{1,\ldots,4}} (\mynext^i \uparrow_1 \wedge \mynext^j \uparrow_1) \big)  \\
& &  \vee \mynext^3 \downarrow_1 \big(\bigvee_{i \neq j \in \set{1,2,3}} (\mynext^i \uparrow_1 \wedge \mynext^j \uparrow_1) \big)  \\
& &  \vee \mynext^4 \downarrow_1 \big(\bigvee_{i \neq j \in \set{1,2}} (\mynext^i \uparrow_1 \wedge \mynext^j \uparrow_1) \big) \big) \\
& &  \\
\aformulater_{0 \sim 6} & = & \downarrow_1( \mynext^6 \uparrow_1)
\end{array}
$$
We write {\rm STA} to denote the formula $\aformulater_{\neg 3/7} \wedge \aformulater_{0 \sim 6}$.

Let $\arun$ be a run of $\aautomaton_P$ and $j$ be such that $0 \leq j <  \length{\arun}$.
We show that  (1) $\arun,j \models {\rm STA}$ iff  (2) ($\arun,j \models \aloc$ for some $\aloc \in \locs$
and $j+6 < \length{\arun}$). In the sequel, we assume that  $j+6 < \length{\arun}$
since otherwise it is clear that  $\arun,j \not \models {\rm STA}$.
By construction, it is clear that (2) implies (1). In order to prove that (1) implies (2), 
we show that if $\arun,j \models \aloc$ for some $\aloc \in \locs_P \setminus \locs$ and
 $j+6 < \length{\arun}$, then  $\arun,j \not \models {\rm STA}$. We perform a systematic case analysis
according to the type of $\aloc$ (we group the cases that require similar arguments):  
%% Ste 081008
%% a positive integer such that $o \leq j \leq \length{\arun} -8-2(m+1)$, we show that 
%% $\arun,j \models {\rm LOC}$ iff $\arun,j \models \aloc$ for $\aloc \in \locs$. By construction, 
%% we have that if $\arun,j \models \aloc$ then $\arun,j \models {\rm LOC}$; in fact, in the first pick of 
%% height $3$ connected to the location $\aloc$, which is composed of $7$ locations, the same counter value will 
%% not be encountered $3$ times, and furthermore the counter value at the 6th next position is effectively equal to 
%% the counter value which is associated to the location $\aloc$. We still have to prove that if $\arun,j 
%% \models {\rm LOC}$ then  $\arun,j \models \aloc$ for $\aloc \in \locs$. We suppose that $\arun,j \models {\rm LOC}$ 
%% and that $\arun,j \models \aloc $ with $\aloc \in \locs_P$ and prove with a case analysis that $\aloc \in \locs$
\begin{enumerate}
\itemsep 0 cm

\item If $\aloc$ is of the form $\aloc^2_i$ with 
      $i \in \interval{2}{m}$, then $\arun,j \not \models \aformulater_{0 \sim 6}$.
      When $\aloc$ is $\aloc_1^2$, $\arun,j \not \models \aformulater_{\neg 3/7}$.

\item If $\aloc$ is of the form $\aloc^3_i$ with 
      $i \in \interval{1}{m}$, then $\arun,j \not \models \aformulater_{0 \sim 6}$.

\item If $\aloc$ is of the form $\aloc^4_i$ with 
      $i \in \interval{1}{m} \setminus \set{2}$, then $\arun,j \not \models \aformulater_{0 \sim 6}$.
      When $\aloc$ is $\aloc_2^4$, $\arun,j \not \models \aformulater_{\neg 3/7}$.

\item  If $\aloc$ is of the form $\aloc_{i,i}$ with 
      $i \in \interval{2}{m-1}$, then $\arun,j \not \models \aformulater_{0 \sim 6}$.
      When $\aloc$ is $\aloc_{m,m}$ and an incrementation is performed after 
      $\aloc_{m,F}$,  we have $\arun,j \not \models \aformulater_{\neg 3/7}$.
      If another action is performed, then we also have $\arun,j \not \models \aformulater_{0 \sim 6}$.

%% Ste 081008
%%
%% \item if $\aloc$ is equal to $\aloc^2_{i},\aloc^3_{i},\aloc^4_{i},\aloc_{i,i}$ with 
%%   $i \in \interval{1}{m}$, we have 
%% $\arun,j \not \models \aformulater_{0 \sim 6}$,
%%

\item If $\aloc$ is of the form either $\aloc^1_i$ or $\aloc^5_i$ with 
       $i \in \interval{1}{m}$, then $\arun,j \not \models \aformulater_{\neg 3/7}$.
%% 
%% \item if $\aloc$ is equal to $\aloc^1_{i}$ with $i \in \interval{1}{m}$, we have $\arun,j \not \models
%%  \aformulater_{\neg 3/7}$; in fact, the counter value in $\aloc^1_{i}$, the one in $\aloc^5_{i}$ and the one in 
%% $\aloc_{i,1}$ are equal,
%%
%% 
%% \item if $\aloc$ is equal to $\aloc^5_{i}$ with $i \in \interval{1}{m}$, we have $\arun,j \not \models 
%% \aformulater_{\neg 3/7}$; in fact the counter value in $\aloc^5_{i}$, the one in $\aloc_{i,1}$ and the one in 
%% $\aloc_{i,2}$ are equal,
%%

\item If $\aloc$ is of the form either $\aloc^0_{i,i}$ or  $\aloc^1_{i,i}$ with 
       $i \in \interval{1}{m}$, then $\arun,j \not \models \aformulater_{\neg 3/7}$.

%%
%% \item if $\aloc$ is equal to $\aloc^0_{i,i}$ with $i \in \interval{1}{m-1}$, we have 
%% $\arun,j \not \models \aformulater_{\neg 3/7}$; in fact the counter value in $\aloc^0_{i,1}$, the one in 
%% $\aloc'_{i,i+1}$ and the one in $\aloc'_{i,i+2}$ are equal,
%% 
%% \item if $\aloc$ is equal to $\aloc^0_{m,m}$, we have $\arun,j \not \models \aformulater_{\neg 3/7}$; in fact the 
%% counter value in $\aloc^0_{m,m}$, the one in $\aloc'_{m,m+1}$ and the one in $\aloc_{i,F}$ are equal,
%% \item if $\aloc$ is equal to $\aloc^1_{i,i}$ with $i \in \interval{1}{m}$, we have $\arun,j \not \models
%%  \aformulater_{\neg 3/7}$; in fact the counter value in $\aloc^1_{i,i}$, the one in $\aloc^3_{i,i}$ and the one 
%% in $\aloc_{i,i+1}$ are equal,
%% 

\item If $\aloc$ is of the form $\aloc^2_{i,i}$ with 
       $i \in \interval{1}{m}$, then $\arun,j \not \models \aformulater_{\neg 3/7}$ (the case $i = m$ requires
      a careful analysis).

%%
%% \item if $\aloc$ is equal to $\aloc^2_{i,i}$ with $i \in \interval{1}{m-1}$, we have $\arun,j \not \models 
%% \aformulater_{\neg 3/7}$; in fact the counter value in $\aloc^2_{i,i}$, the one in $\aloc_{i,i+1}$ and the one 
%% in $\aloc_{i,i+2}$ are equal,
%% 
%% 
%% \item if $\aloc$ is equal to $\aloc^2_{m,m}$, we have $\arun,j \not \models \aformulater_{\neg 3/7} \wedge 
%% \aformulater_{0 \sim 6}$; in fact if $\arun,j \models \aformulater_{0 \sim 6} $ since the counter value in 
%% $\aloc^2_{m,m}$ and the one in $\aloc_{m,m+1}$ are equal, if the one at the 6th next position is also equal to 
%% them, we have $\arun,j \not \models \aformulater_{\neg 3/7}$,
%% 

\item If $\aloc$ is of the form  $\aloc_{i,k}$
      for some  $i \in \interval{1}{m}$, $k \in \interval{1}{m-1}$ such that
      either $\length{i-k} > 2$ or $k > i$, then $\arun,j \not \models \aformulater_{\neg 3/7}$.
%%
%%\item if $\aloc$ is equal to $\aloc_{i,k}$ with $i \in \interval{1}{m}$ and $k \in \interval{1}{m-1}$ and 
%% either $|(i-k)| > 2$ either $k > i$, we have $\arun,j \not \models \aformulater_{\neg 3/7}$; in fact the counter 
%% value in $\aloc_{i,k}$, the one in $\aloc_{i,k+1}$ and the one in $\aloc_{i,k+2}$ are equal,
%%

\item If $\aloc$ is of the form $\aloc_{i,i-1}$ with $i \in \interval{2}{m}$, then 
      $\arun,j \not \models \aformulater_{\neg 3/7}$.
%% 
%% \item if $\aloc$ is equal to $\aloc_{i,i-1}$ with $i \in \interval{2}{m}$, we have 
%% $\arun,j \not \models \aformulater_{\neg 3/7}$; in fact the counter value in $\aloc_{i,i-1}$, the one in 
%% $\aloc^1_{i,i}$ and the one in $\aloc^2_{i,i}$ are equal,
%% 

\item If $\aloc$ is of the form $\aloc_{i,i-2}$ with $i \in \interval{3}{m}$, then 
      $\arun,j \not \models \aformulater_{\neg 3/7}$.

%% 
%% \item if $\aloc$ is equal to $\aloc_{i,i-2}$ with $i \in \interval{3}{m}$, we have $\arun,j \not \models 
%% \aformulater_{\neg 3/7}$; in fact the counter value in $\aloc_{i,i-2}$, the one in $\aloc_{i,i-1}$ and the one 
%% in $\aloc^1_{i,i}$ are equal,
%%

\item If $\aloc$ is of the form $\aloc_{i,m}$ with $i \in \interval{1}{m-1}$, then 
      $\arun,j \not \models \aformulater_{\neg 3/7}$.

%%
%% \item if $\aloc$ is equal to $\aloc_{i,m}$ with $i \in \interval{1}{m-1}$, we have $\arun,j \not \models 
%% \aformulater_{\neg 3/7} \wedge \aformulater_{0 \sim 6}$; in fact if $\arun,j \models \aformulater_{0 \sim 6} $ 
%% since the counter value in $\aloc_{i,m}$ and the one in $\aloc_{i,m+1}$ are equal, if the one at the 6th next 
%% position is also equal to them, we have $\arun,j \not \models \aformulater_{\neg 3/7}$,
%% 

\item If $\aloc$ is of the form  $\aloc_{i,m+1}$ with $i \in \interval{1}{m}$ 
      and an action different from decrementation is performed after $\aloc_{i,F}$, then
      $\arun,j \not \models \aformulater_{0 \sim 6}$.
      When a decrementation is performed after $\aloc_{i,F}$, we get 
       $\arun,j \models \aformulater_{0 \sim 6} \wedge \neg \aformulater_{\neg 3/7}$.
      
%%
%%
%% \item if $\aloc$ is equal to $\aloc_{i,m+1}$ with $i \in \interval{1}{m}$,  $\arun,j \not \models 
%% \aformulater_{\neg 3/7} \wedge \aformulater_{0 \sim 6}$; in fact $\arun,j \models \aformulater_{0 \sim 6} $, 
%% this means that there must be a transition of the form $(\aloc_{i,F},\dec,\aloc_k)$ which is fired, hence the 
%% counter value in $\aloc^4_k$  , which is the one at the 6th next position is then equal to the one in 
%% $\aloc_{i,m+1}$ and also to the one in $\aloc^2_k$, consequently $\arun,j \not \models \aformulater_{\neg 3/7}$,
%%
%%

\item If $\aloc$ is of the form  $\aloc'_{i,k}$
      for some  $i \in \interval{1}{m}$, $k \in \interval{1}{m-1}$ such that
      either $\length{i-k} > 2$ or $k > i$, then $\arun,j \not \models \aformulater_{\neg 3/7}$.

%%
%% \item if $\aloc$ is equal to $\aloc'_{i,k}$ with $i \in \interval{1}{m}$ and $k \in \interval{1}{m-1}$ and 
%% either $|(i-k)| > 2$ either $k > i$, we have $\arun,j \not \models \aformulater_{\neg 3/7}$; in fact the counter 
%% value in $\aloc'_{i,k}$, the one in $\aloc'_{i,k+1}$ and the one in $\aloc'_{i,k+2}$ are equal,
%%

\item If $\aloc$ is of the form $\aloc'_{i,i-1}$ with $i \in \interval{2}{m}$, then 
      $\arun,j \not \models \aformulater_{\neg 3/7}$.

%%
%% \item if $\aloc$ is equal to $\aloc'_{i,i-1}$ with $i \in \interval{2}{m}$, we have $\arun,j \not \models
%%  \aformulater_{\neg 3/7}$; in fact the counter value in $\aloc'_{i,i-1}$, the one in $\aloc^0_{i,i}$ and the one 
%% in $\aloc'_{i,i+1}$ are equal,
%% 

\item If $\aloc$ is of the form $\aloc'_{i,i-2}$ with $i \in \interval{3}{m}$, then 
      $\arun,j \not \models \aformulater_{\neg 3/7}$.

%% 
%% \item if $\aloc$ is equal to $\aloc'_{i,i-2}$ with $i \in \interval{3}{m}$, we have $\arun,j \not \models 
%% \aformulater_{\neg 3/7}$; in fact the counter value in $\aloc'_{i,i-2}$, the one in $\aloc'_{i,i-1}$ and the one 
%% in $\aloc^0_{i,i}$ are equal,
%%

\item  If $\aloc$ is of the form $\aloc'_{i,m}$ with $i \in \interval{1}{m}$, then 
      $\arun,j \not \models \aformulater_{\neg 3/7}$.

%%
%% \item if $\aloc$ is equal to $\aloc'_{i,m}$ with $i \in \interval{1}{m}$, we have $\arun,j \not \models 
%% \aformulater_{\neg 3/7}$; in fact the counter value in $\aloc'_{i,m}$, the one in $\aloc'_{i,,+1}$ and the one in 
%% $\aloc'_{i,F}$ are equal,
%%

\item If $\aloc$ is of the form  $\aloc'_{i,m+1}$ with $i \in \interval{1}{m}$,
then $\arun,j \not \models \aformulater_{0 \sim 6}$. Indeed, the 6th next position, if any,
is of the form $\aloc^3_{k}$ for some $k \in \interval{1}{m}$. The counter
value at such a position is strictly greater than the one at the position $j$ whatever is the action
performed after $\aloc_{i,F}$.

%% 
%% \item if $\aloc$ is equal to $\aloc'_{i,m+1}$ with $i \in \interval{1}{m-1}$, we have $\arun,j \not \models 
%% \aformulater_{0 \sim 6}$; in fact if $\arun,j \models \aformulater_{0 \sim 6} $, the location at the 6th next 
%% position is $\aloc^3_k$ for a $k \in \interval{1}{m}$ and the associated counter value cannot be equal to the one 
%% in $\aloc'_{i,m+1}$ whatever is the form of the transition between $\aloc_{i,F}$ and $\aloc_k$,
%% 

\item If $\aloc$ is of the form $\aloc_{i,F}$ with $i \in \interval{1}{m}$ and
the action performed after $\aloc_{i,F}$ is not a decrementation, then 
$\arun,j \not \models \aformulater_{0 \sim 6}$.
When a decrementation is performed after $\aloc_{i,F}$, we get 
       $\arun,j \models \aformulater_{0 \sim 6} \wedge \neg \aformulater_{\neg 3/7}$.

%%
%%\item if $\aloc$ is equal to $\aloc_{i,F}$ with $i \in \interval{1}{m}$, we have  $\arun,j \not 
%% \models\aformulater_{\neg 3/7} \wedge \aformulater_{0 \sim 6}$; in fact if $\arun,j \models 
%% \aformulater_{0 \sim 6} $, this means that there must be a transition of the form $(\aloc_{i,F},\dec,\aloc_k)$ 
%% which is fired, hence the counter value in $\aloc^5_k$  , which is at the 6th next position, is equal to the one 
%% in $\aloc_{i,F}$ and is also equal to the one in $\aloc^1_k$, and consequently $\arun,j \not \models
%%  \aformulater_{\neg 3/7}$.
\end{enumerate}

For $i \in \interval{1}{m}$, let us define the formula 
 $\aformula_i= \mynext^{6 + 2(i-1)} \downarrow_1 \mynext^2 \neg \uparrow_1$. 
One can check that in the run of $\aautomaton_P$,  ${\rm STA} \wedge \aformula_i$ holds true iff 
the current state is $\aloc_i$ and there are at least 6 following positions.

Let $\aformula$ be a formula in $\fLTL{\downarrow,\locs}{n}$.  We define $\aformula_P$ as the formula 
$\atranslation(\aformula)$ such that the map
$\atranslation$ is homomorphic for Boolean operators and $\downarrow_r$, and its restriction to 
$\uparrow_r$ is identity. 
The rest of the inductive definition  is as follows. 
\begin{itemize}
\itemsep 0 cm
\item $\atranslation(\aloc_i)= \aformula_i$,
\item $\atranslation(\mynext \aformula) = \mynext^{9+2(m+1)+1} \atranslation(\aformula)$,
\item $\atranslation(\aformula \until \aformula')=
       \big ({\rm STA} \Rightarrow \atranslation(\aformula) \big) 
\until \big ({\rm STA} \wedge \atranslation(\aformula') \big )$.
\end{itemize}
Observe that $\aformula$ and $\aformula_P$ have the same amount of registers unless 
$\aformula$ has no register. For each accepting run in $\aautomaton$, there exists an accepting run in 
$\aautomaton_P$ and conversely for each accepting run  in $\aautomaton_P$, there exists an accepting run in 
$\aautomaton$. Furthermore the sequence of  counter values for the configurations of each of these runs which have a 
state in $\locs$ match.
\hfill $\Box$
\end{proof}

\begin{lemma}[{\small Purification for $\fo{}{\sim,<,+1}$}]
\label{lemma-purification-fo}
Given a one-counter automaton $\aautomaton$
              and  an $\fo{\locs}{\sim,<,+1}$ sentence $\aformula$
              with $n$ variables, one can compute in logarithmic space in
$\length{\aautomaton} + \length{\aformula}$ a one-counter automaton 
$\aautomaton_P$ and $\aformula_P$ in $\fo{\emptyset}{\sim,<,+1}$ with at most $n+2$ variables
such that  $\aautomaton \models^{\fin} \aformula$
                 [resp.\ $\aautomaton \models^{\infin} \aformula$] iff 
$\aautomaton_P \models^{\fin} \aformula_P$
                 [resp.\ $\aautomaton_P \models^{\infin} \aformula_P$].
Moreover, $\aautomaton$ is deterministic 
% [resp.\ weakly deterministic] 
iff $\aautomaton_P$  is deterministic.
% [resp.\ weakly deterministic]. 
\end{lemma}

\begin{proof} 
The proof follows the lines of the proof of Lemma~\ref{lemma-purification} by considering the first-order formulae corresponding to 
the formulae {\rm STA} and $\aformula_i$ and the same automaton construction. In order to make this construction feasible, we need to use formulae of the 
form  $\avariable = \avariablebis  + k$. In fact, the formulae of the form $\avariable = \avariablebis  + 1$ are translated into 
formulae of the form $\avariable = \avariablebis  + 9 + 2(m+1)$ (this case is identical to the case of the formulae of the form $\mynext \aformula$). 
Typically, encoding $\avariable = \avariablebis  + k$ for the constant $k$ requires two auxiliary variables. 
For instance we can encode the formula $\avariable=\avariablebis +4$ as follows:
$$
\exists \ \avariablebis_2 \ \avariable=\avariablebis_2 +1 \wedge (\exists \ \avariablebis_1 \ \avariablebis_2=\avariablebis_1+1 \wedge (\exists \ \avariablebis_2 \ \avariablebis_1=\avariablebis_2+1 \wedge \avariablebis_2=\avariablebis+1) )
$$
Here again, we recycle the variables $\avariablebis_1$ and $\avariablebis_2$.
\hfill $\Box$
\end{proof}

%%
%% Ste 150807
%% material in the file mc.tex are currently included in the file preliminaries.tex
%% \section{Model-checking one-counter machine}
%% \input{mc}
%% %
%5 %
\section{Model checking  deterministic one-counter automata}
\label{section-decidability} 

In this section, we show that $\fMC{\fin}{}$ and  $\fMC{\infin}{}$ restricted to deterministic one-counter automata is \pspace-complete.
%% Arn 170308
% In this section, we show that $\fMC{\omega}{}$ restricted to
% deterministic one-counter automata is \pspace-complete and the same
% restriction for $\fMC{< \omega}{}$ is in \expspace.
%% 
%% Actually, this still holds true by addition of a finite amount of MSO-definable temporal
%% operators. 
%% 
%%
\iffossacs
First, we show \pspace-hardness.
\else
\subsection{\pspace \  lower bound}
\label{section-pspace-lower-bound}
We show below a \pspace-hardness result by taking advantage of the alphabet of states by means of a reduction from \QBF \ (``Quantified Boolean Formula'') that is a standard
\pspace-complete problem. 
\fi
%%
%% The Purification Lemma will then allows us to get rid of locations
%% and to specify only constraints on the counter values. 
%%
\begin{proposition} \label{proposition-pspace-hardness} 
$\fpureMC{\fin}{}$ and 
$\fpureMC{\infin}{}$ restricted to deterministic one-coun\-ter automata are \pspace-hard problems. Furthermore, for $\fpureMC{\fin}{}$  [resp. $\fpureMC{\infin}{}$] this results holds for formulae  using only the temporal operators $\mynext$ and $\sometimes$ [resp. $\sometimes$].
\end{proposition}

\iffossacs
\else
\begin{proof}
 Consider a \QBF \ instance $\aformula$:
$
\aformula = 
\forall \avarprop_1 \
\exists \avarprop_2 \ \cdots \
\forall \avarprop_{2N-1} \ \exists \ \avarprop_{2N} \ \Psi(\avarprop_1,...,\avarprop_{2N})
$
where  $\avarprop_1$,...,$\avarprop_{2N}$ are propositional variables and 
$\Psi(\avarprop_1, \ldots,\avarprop_{2N})$ is a quantifier-free propositional 
formula built over $\avarprop_1, \ldots, \avarprop_{2N}$.
The  fixed deterministic one-counter automaton
$\aautomaton$ below generates the sequence of counter values $(0 1)^{\omega}$. 
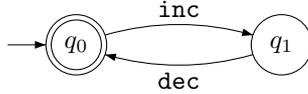
\begin{center}
\begin{picture}(30,10)(0,0)
\setlength{\unitlength}{0.9mm}
\node[Nmarks=ir](Q0)(0,5){$\aloc_0$}
\node(Q1)(30,5){$\aloc_1$}
\drawedge[curvedepth=3](Q0,Q1){$\mathtt{inc}$}
\drawedge[curvedepth=3](Q1,Q0){$\mathtt{dec}$}
\end{picture}
\end{center}
Let $\aformulabis$ be the formula in $\fLTL{\downarrow,\emptyset}{}$
defined from the family $\aformulabis_1, \ldots,
\aformulabis_{2N+1}$ of formulae with $\aformulabis = \downarrow_{2N+1} \aformulabis_1$. 
\begin{itemize}
\itemsep 0 cm
\item $\aformulabis_{2N+1} = \Psi(\uparrow_1 \Leftrightarrow \uparrow_{2N+1},\ldots,\uparrow_{2N} \Leftrightarrow \uparrow_{2N+1})$,
\item for $i \in \set{1,...,N}$, $\aformulabis_{2i} = \sometimes (\downarrow_{2i} \aformulabis_{2i+1})$ and
      $\aformulabis_{2i-1} = \always (\downarrow_{2i-1} \aformulabis_{2i})$. 
\end{itemize}
One can show that  $\aformula$ is satisfiable 
%% iff $\aautomaton_{\aformula} \models_{< \omega} \aformulabis$ 
iff $\aautomaton \models^{\infin} \aformulabis$. 

To do so, we proceed as follows. For $i \in \set{0,2,4,6,\ldots,2N}$, let $\aformula_i$ be
$$
\aformula_i = 
\forall \avarprop_{i+1} \
\exists \avarprop_{i+2} \ \cdots \
\forall \avarprop_{2N-1} \ \exists \ \avarprop_{2N} \ \Psi(\avarprop_1,...,\avarprop_{2N}).
$$
So $\aformula_0$ is precisely $\aformula$. 
Similarly, for 
 $i \in \set{1,3,5,\ldots,2N-1}$, let $\aformula_i$ be
$$
\aformula_i = 
\exists \avarprop_{i+1} \
\forall \avarprop_{i+2} \ \cdots \
\forall \avarprop_{2N-1} \ \exists \ \avarprop_{2N} \ \Psi(\avarprop_1,...,\avarprop_{2N}).
$$
Observe that the free propositional variables in $\aformula_i$ are exactly $\avarprop_1,
\ldots, \avarprop_{i}$ and $\aformula_i$ is obtained from $\aformula$ by removing the $i$ first 
quantifications. 
Given a propositional valuation $\aregval: \set{\avarprop_1, \ldots,
\avarprop_i} \rightarrow \set{\top,\perp}$ for some $i \in \set{1, \ldots,2N}$,
we write $\overline{\aregval}$ to denote a register valuation such that its restriction
to $\set{1,\ldots,i,2N+1}$ satisfies: $\aregval(\avarprop_j) = \top$ iff $\overline{\aregval}(j) = 
0$ for $j \in \set{1,\ldots,i}$ and $\overline{\aregval}(2N+1) = 0$. One can show by induction that 
for $k \geq 0$, $\aregval \models \aformula_{i-1}$ (in \QBF) iff
$\arun_{\aautomaton}^{\infin}, k \models_{\overline{\aregval}}  \aformulabis_{i}$, where $\arun_{\aautomaton}^{\omega}$ denotes
the unique infinite run for $\aautomaton$.
Consequently, if $\aregval \models \aformula$ for some propositional valuation, then
$\arun_{\aautomaton}^{\omega}, 0 \models_{\overline{\aregval}}  \aformulabis$.
Similarly, if $\arun_{\aautomaton}^{\omega}, 0 \models_{\aregval}  \aformulabis$, then
there is a propositional valuation $\aregval'$ such that
$\overline{\aregval'} = \aregval$ and $\aregval' \models \aformula$.

% For $\fpureMC{< \omega}{}$, one can enforce the sequence of counter values from the accepting run to be $(01)^{2N}0$ and then use $\mynext$ to define the $\aformulabis_{i}$s.
For the finitary problem $\fpureMC{\fin}{}$, the above proof does not work because the 
occurrences of $\always$ related to universal quantification in the QBF formula
might lead to the end of the run, leaving no choice for the next  quantifications. Consequently, one need to use another 
deterministic one-counter automaton with $4N+1$ states such that the sequence of counter values from the accepting run is 
$(01)^{2N}0$ (again we omit useless $\mathtt{ifzero}$ transitions). 
Let us consider  the deterministic counter automaton $\aautomaton'$ below.

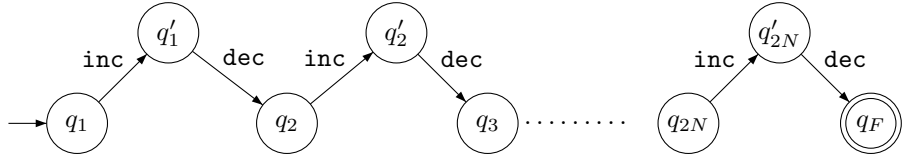
\begin{center}
\begin{picture}(90,20)(0,0)
{\setlength{\unitlength}{1.2mm}
\node[Nmarks=i](Q0)(0,5){$\aloc_1$}
\node(Q01)(10,15){$\aloc'_1$}
\drawedge(Q0,Q01){$\inc$}
\node(Q1)(23,5){$\aloc_2$}
\drawedge[ELside=l](Q01,Q1){$\dec$}
\node(Q11)(35,15){$\aloc'_2$}
\drawedge(Q1,Q11){$\inc$}
\node(Q2)(45,5){$\aloc_3$}
\drawedge[ELside=l](Q11,Q2){$\dec$}
\put(49,5){$\ldots \ldots \ldots$}
\node(QN)(67,5){$\aloc_{2N}$}
\node(QN1)(77,15){$\aloc'_{2N}$}
\drawedge(QN,QN1){$\inc$}
\node[Nmarks=r](QF)(87,5){$\aloc_{F}$}
\drawedge[ELside=l](QN1,QF){$\dec$}
}
\end{picture}
\end{center}

We shall build another formula $\aformulabis$ in $\fLTL{\downarrow,\emptyset}{}$ defined from the 
formulae below with $\aformulabis = \downarrow_{2N+1} \aformulabis_1$. 
\begin{itemize}
\itemsep 0 cm
\item $\aformulabis_{2N+1} = \Psi(\uparrow_1 \Leftrightarrow \uparrow_{2N+1},\ldots,\uparrow_{2N} \Leftrightarrow \uparrow_{2N+1})$,
\item for $i \in \set{1,...,N}$:
\begin{itemize}
\itemsep 0 cm
\item[\bf{--}] $\aformulabis_{2i} = \sometimes \big( (\mynext^{4N-4i+2}\ \top) \wedge \downarrow_{2i} \aformulabis_{2i+1} \big)$ and
\item[\bf{--}]   $\aformulabis_{2i-1} = \always \big ((\mynext^{4N-4i+4} \ \top) \Rightarrow \downarrow_{2i-1} \aformulabis_{2i}\big )$. 
\end{itemize}
Herein, $\top$ holds for the truth value that can be encoded with $\downarrow_1 \vee \neg \downarrow_1$ (remember there
are no propositional variables in the pure version of the model-checking problems). 
\end{itemize}
Using a similar proof by induction as the one done for the infinite case, we obtain that  $\aformula$ is satisfiable  iff $\aautomaton' \models^{\fin} \aformulabis$. 
\hfill $\Box$
\end{proof}
\fi

Observe that in the reduction for $\fpureMC{\omega}{}$, we use an unbounded number of registers 
(see Theorem~\ref{theorem-det}) but a fixed deterministic
one-counter automaton.

By Lemmas~\ref{lemma-purification-fo} and~\ref{lemma-ltl-fo}, we obtain the following corollary.

\begin{corollary}
$\fopureMC{\fin}{}$ and 
$\fopureMC{\infin}{}$ restricted to deterministic one-coun\-ter automata are \pspace-hard problems.
\end{corollary}

\subsection{Properties on runs for deterministic automata}
\label{section-properties}

Any deterministic one-counter automaton $\aautomaton$ 
has at most one infinite run, possibly with an infinite amount of counter values. 
If this run is not accepting,
i.e. no accepting state is repeated infinitely often, then 
for no formula $\aformula$, we have $\aautomaton \models^{\omega} \aformula$. 
We show below that we can decide in polynomial-time whether $\aautomaton$ has
accepting runs either finite or infinite. Moreover, we shall show that the infinite unique
run has some regularity.

Let $\arun_{\aautomaton}^{\omega}$ be the unique infinite run (if it exists) of the deterministic one-counter
automaton $\aautomaton$ represented by the following sequence of configurations
$$
\pair{\aloc_0}{n_0} \ \pair{\aloc_1}{n_1} \ \pair{\aloc_2}{n_2} \ldots
$$
%%
%% Ste 051007
%%
%% We write ${\rm ZERO}(\aautomaton)$ to denote the set of positions of
%% $\arun_{\aautomaton}^{\omega}$ where a zero-test has been successful.
%% By convention, $0$ belongs to ${\rm ZERO}(\aautomaton)$ since in a run
%% we require that the first configuration is the initial configuration of
%% $\aautomaton$ with counter value $0$. Hence, $ {\rm ZERO}(\aautomaton)
%% \egdef \set{0} \cup \set{i > 0: n_i = n_{i+1} = 0} $.
%% 
%% 
%% 
%% \begin{lemma} \label{lemma-between-zero-tests}
%%  Let $i < j$ be in ${\rm ZERO}(\aautomaton)$ for which there is no $i < k < j$
%% with $k \in {\rm ZERO}(\aautomaton)$. Then, $(j - i) \leq \length{\locs}^2$.
%% \end{lemma}
%% 
%% \iffossacs
%% The proof (see Appendix~\ref{section-proofs-section-properties})
%%  essentially establishes that the counter cannot go beyond $\length{\locs}$ 
%% between two positions with successful zero-tests. 
%% \else
%% \begin{proof} 
%% \input{proof-lemma-between-zero-tests}
%% \qed
%% \end{proof}
%% \fi
%% 

Lemma~\ref{lemma-KKK} below is a key result to show the forthcoming \pspace \ upper bound.
Basically, the unique run of deterministic one-counter automata  has regularities that can be described
in polynomial size. 

\begin{lemma} \label{lemma-KKK}
Let $\aautomaton$ be a deterministic one-counter automaton with an infinite  run. 
There are $K_1, K_2, K_{inc}$ such that $K_1 + K_2 \leq \length{\locs}^3$,
$K_{inc} \leq \length{\locs}$ and 
for every $i \geq K_1$, $\pair{\aloc_{i + K_2}}{n_{i+K_2}} = \pair{\aloc_i}{n_i + K_{inc}}$.
%% Ste 051007
%% Moreover, if ${\rm ZERO}(\aautomaton)$ is infinite, then $K_{inc} = 0$, otherwise
%% $K_2 \leq \length{\locs}$.
\end{lemma}

Hence, the run $\arun_{\aautomaton}^{\omega}$ can be  encoded by its first $K_1 + K_2$
configurations. It is worth noting that we have deliberately decided to keep the three constants
$K_1$, $K_2$ and $K_{inc}$ in order to provide a more explicit analysis. 
\iffossacs
%See the proof in Appendix~\ref{section-proofs-section-properties}. 
\else
\begin{proof} (Lemma~\ref{lemma-KKK})
We write ${\rm ZERO}(\aautomaton)$ to denote the set of positions of
$\arun_{\aautomaton}^{\omega}$ where a zero-test has been successful.
By convention, $0$ belongs to ${\rm ZERO}(\aautomaton)$ since in a run
we require that the first configuration is the initial configuration of
$\aautomaton$ with counter value $0$. Hence, $ {\rm ZERO}(\aautomaton)
\egdef \set{0} \cup \set{i > 0: n_i = n_{i+1} = 0} $.
Let us first establish Lemma~\ref{lemma-between-zero-tests} below.

\begin{lemma} \label{lemma-between-zero-tests}
 Let $i < j$ be in ${\rm ZERO}(\aautomaton)$ for which there is no $i < k < j$
with $k \in {\rm ZERO}(\aautomaton)$. Then, $(j - i) \leq \length{\locs}^2$.
\end{lemma}

The proof essentially establishes that the counter cannot go beyond $\length{\locs}$ 
between two positions with successful zero-tests. 

\begin{proof} (Lemma~\ref{lemma-between-zero-tests})
First observe that there are no $i < k < k' < j$ such that
  $\aloc_k = \aloc_{k'}$ and $n_k \leq n_{k'}$. Indeed, if it is the
  case since there is no successful zero-tests in 
$\pair{\aloc_{i+1}}{n_{i+1}}
  \cdots \pair{\aloc_k}{n_k}  \linebreak[0] \cdots \pair{\aloc_{k'}}{n_{k'}}$
 and
  $\aautomaton$ is deterministic we would obtain from
  $\pair{\aloc_{k'}}{n_{k'}}$ an infinite path with no zero-test, a
  contradiction with the existence of $\pair{\aloc_j}{n_j}$.  Hence,
  if there are $i < k < k' < j$ such that $\aloc_k = \aloc_{k'}$, then
  $n_{k'} < n_{k}$.  Now suppose that there is $i < k < j$ such that
  $n_k \geq \length{Q}$.  We can extract a subsequence
  $\pair{\aloc_{i_0}}{n_{i_0}} \cdots
  \pair{\aloc_{i_{s}}}{n_{i_{s}}}$ from $\pair{\aloc_{i}}{n_{i}}
  \cdots \pair{\aloc_{n_k}}{n_k}$ such that $i_0 = i$, $i_{s} = k$
  and for $0 \leq l < s$, $n_{i_{l+1}} = n_{i_l} + 1$.
  Consequently, there are $l,l'$ such that $\aloc_{i_l} =
  \aloc_{i_{l'}}$ and $n_{i_l} < n_{i_{l'}}$, which leads to a
  contradiction from the above point.  Hence, for $ k \in \set{i,
    \ldots,j}$, $n_k \leq \length{Q} - 1$.  Since $\aautomaton$ is
  deterministic, this implies that 
  %% Ste 061008
  %% $(j - i) \leq \length{\locs} \times (\length{\locs} -1)$.
  $(j - i) \leq \length{\locs} \times \length{\locs}$.

\hfill $\Box$
\end{proof}

Let us come back to the rest of the proof.

First, suppose that ${\rm ZERO}(\aautomaton)$ is infinite. 
  Let $i_0 < i_1 < i_2 < \ldots$ be the infinite sequence composed of
  elements from ${\rm ZERO}(\aautomaton)$ ($i_0 = 0$). There are $l,l'
  \leq \length{\locs}$ such that $\pair{\aloc_{i_l}}{n_{i_l}} =
  \pair{\aloc_{i_{l'}}}{n_{i_{l'}}}$. By
  Lemma~\ref{lemma-between-zero-tests}, $i_{l'} \leq \length{\locs}
  \times \length{\locs}^2$ . Take $K_1 = i_l$ and $K_2 = i_{l'} -
  i_{l}$. 
  
  Second, suppose that ${\rm ZERO}(\aautomaton)$ is finite, say equal
  to $\set{0, i_1, \ldots, i_l}$ for some $l \leq \length{\locs} - 1$
  (if $l \geq \length{\locs}$ we are in the first case).  By
  Lemma~\ref{lemma-between-zero-tests}, $i_{l} \leq (\length{\locs}-1)
  \times \length{\locs}^2$.  For all $i_l \leq k < k'$, if $\aloc_k =
  \aloc_{k'}$, then $n_{k} \leq n_{k'}$ (if it were not the case, there
  would eventually be another zero-test in the path starting with
  $\pair{\aloc_{i_l}}{n_{i_l}}$). Now there are $i_l \leq k < k' \leq
  i_l + \length{\locs}$ such that $\aloc_k = \aloc_{k'}$ and
  consequently $n_{k} \leq n_{k'}$.  Take $K_1 = k$, $K_2 = k' - k$
  and $K_{inc} = n_{k'} - n_{k}$. We have $K_{inc} \leq \length{\locs}$ because $k'-k \leq \length{\locs}$. 
\hfill $\Box$\\
\end{proof}

\fi 
$\arun_{\aautomaton}^{\omega}$ has a simple structure: it is composed
of a polynomial-size prefix 
$$\pair{\aloc_0}{n_0} \cdots \linebreak[0] \pair{\aloc_{K_1-1}}{n_{K_1-1}}
$$
followed by the polynomial-size loop  
$\pair{\aloc_{K_1}}{n_{K_1}} \cdots \linebreak[0] \pair{\aloc_{K_1+ K_2-1}}{n_{K_1+K_2-1}}$
repeated infinitely often. The effect of applying the loop consists in adding $K_{inc}$ to every counter
value. Testing whether  
$\aautomaton$ has an infinite run or
$\arun_{\aautomaton}^{\omega}$ is accepting amounts to check whether
there is an accepting state in the loop, which can be done in cubic time
in $\length{\locs}$. 
In the rest of this section, we assume that $\arun_{\aautomaton}^{\omega}$ is accepting. 
Similarly, testing whether $\aautomaton$ has a finite accepting
run amounts to check whether an accepting state occurs in the prefix or in the loop. 

When $K_{inc} = 0$ and $\aautomaton$ has an infinite run,
 $\arun_{\aautomaton}^{\omega}$ 
is exactly
$$\pair{\aloc_0}{n_0} \cdots \pair{\aloc_{K_1-1}}{n_{K_1-1}}
 (\pair{\aloc_{K_1}}{n_{K_1}} \cdots \pair{\aloc_{K_1+ K_2-1}}{n_{K_1+K2-1}})^{\omega}.
$$
It is then possible to apply a polynomial-space labelling algorithm \`a la CTL 
for model checking $\fLTL{\downarrow,\locs}{}$ formulae on $\aautomaton$.
However, one needs  to take care of register valuations, which explains why
unlike the polynomial-time algorithm for model checking ultimately periodic models
on LTL formulae (see e.g., \cite{Markey&Schnoebelen03}),  model checking restricted to deterministic 
automata with $K_{inc} = 0$  is 
still  \pspace-hard (see the proof of Proposition~\ref{proposition-pspace-hardness}).

\subsection{A \pspace \ symbolic model-checking algorithm}
\label{section-pspace}

In this section, we provide decision procedures for solving 
 $\foMC{\fin}{}$ and 
$\foMC{\infin}{}$ restricted to deterministic one-counter automata.
Let us introduce some notations.
Let  $\arun_{\aautomaton}^{\omega} = \pair{\aloc_0}{n_0} \ 
\pair{\aloc_1}{n_1} \linebreak[0]\ \pair{\aloc_2}{n_2} \ldots$ 
be the unique run of the deterministic one-counter
automaton $\aautomaton$.
%% Ste 210408
%%  and $\aformula$ be a sentence with $N \geq 1$ registers.
%% Let $i \geq 0$ be a position in $\arun_{\aautomaton}^{\omega}$
%% and $m$ be a register value in $\Nat$. 
%% We write ${\rm pos}_{\aautomaton}(i,m)$ to denote the following (possibly infinite) set
%% of offsets:
%% $
%% {\rm pos}_{\aautomaton}(i,m) =
%% \set{j \in \Nat: m = n_{i+j}}$. 
%% The values $m$ should be understood as register values when evaluation of subformulae
%% is done at position $i$.
%% In general, the set 
%% $
%% \set{{\rm pos}_{\aautomaton}(i,m) \subseteq \Nat: i,m \in \Nat}
%% $
%% can be infinite but if we restrict ourselves to $m$ in $\set{n_0,\ldots,n_{i}}$ then it is
%% not anymore the case. After all, this is a reasonable assumption 
%% when $m$ is intended to be a value stored in a register.  
%% %% Arn 310807
%% %%
%% %% I introduce here a technical lemma which help (I think) for better
%% %% understanding the proof of the next lemmA
%%
%% Before showing this property, 

We establish
that whenever $K_{inc} > 0$, two positions with identical counter values
are separated by a distance that is bounded by a polynomial in $\length{\locs}$. 

%% Ste 210408
%% \begin{lemma}\label{lemma-bound-context}
%% Suppose $K_{inc} > 0$. For all $i \leq j$, 
%% \iffossacs
%% (I) $n_i = n_j$ and $i < K_1$ imply $(j-i) \leq K_1 + K_1K_2$,
%% (II) $n_i = n_j$ and $i \geq K_1$ imply $(j-i) \leq K_2^2$. 
%% \else
%% \begin{description}
%% \itemsep 0 cm
%% \item[(I)] $n_i = n_j$ and $i < K_1$ imply $(j-i) \leq K_1 + K_1K_2$,
%% \item[(II)]$n_i = n_j$ and $i \geq K_1$ imply $(j-i) \leq K_2^2$. 
%% \end{description}
%% \fi
%% %% 
%% %%   When $K_3>0$, for all $i<K_1$ [resp. for all $i \geq K_1$], for all $j \geq i$, if $n_i=n_j$ then 
%% %% $(j-i) \leq K_1 + K_1K_2$ [resp. $(j-i) \leq K_2^2$].
%% %% 
%% \end{lemma}
%% 
%% 
%% \begin{proof} 
%% \input{proof-lemma-bound-context}
%% \qed
%% \end{proof}
%%
%%  
%%
%%
%% Lemma~\ref{lemma-bound-context} can be slightly refined. To do so, 
Let us introduce a few constants related to the one-counter automaton $\aautomaton$ when $K_{inc} > 0$.
\begin{itemize}
\itemsep 0 cm
\item Let $\beta_1, \beta_2 \geq 0$ be the smallest natural numbers 
      such that for every $i \in [K_1, K_1 + K_2 -1]$,
     $n_i \in [n_{K_1} - \beta_1, n_{K_1} + \beta_2]$.
\item Let $\gamma$ be the greatest value amongst $\set{n_0, \ldots,n_{K_1-1}}$.  
\item $L = 1 + \gamma + \left\lceil  \frac{\beta_1 + \beta_2}{K_{inc}} \right\rceil$ where
$\lceil \cdot \rceil$ denotes the ceiling function.
\end{itemize}

Intuitively, the constant $LK_2$ is greater than any distance between two positions belonging to the loop  of the unique infinite 
run of $\aautomaton$ which have the same counter value. The next lemma formalizes this idea.

\begin{lemma}\label{lemma-bound-L}
 Suppose $K_{inc} > 0$ and let $i,j$ be in $\Nat$. 
\begin{enumerate}
\itemsep 0 cm 
\item If $i,j \geq K_1$ and $|i-j| \geq LK_2$, then $n_i \neq n_j$.
\item If $i < K_1$ and $j \geq K_1 + LK_2$, then $n_i \neq n_j$.
\end{enumerate}
\end{lemma}

\begin{proof}
 (1) Assume that  $i,j \geq K_1$ and $(i-j) \geq LK_2$. 
By using the Euclidean division, we introduce the following values:
$r_i=(i-K_1) \bmod(K_2)$, $r_j=(j-K_1) \bmod(K_2)$ and the quotients $a_i$ and $a_j$ such that
 $i-K_1=a_iK_2 + r_i$ and $j-K_1=a_jK_2 +r_j$. Note that
$0 \leq r_i,r_j < K_2$  and since $(i-j) \geq LK_2$, we necessarily have $a_i-a_j > L-1$. Using the definition of the 
constants $\beta_1$ and $\beta_2$, we know that $n_{r_i+K_1},n_{r_j+K_1} \in \set{n_{K_1}-\beta_1,\ldots,n_{K_1}+\beta_2}$. 
Since $i=a_iK_2 + r_i + K_1$ and $j=a_jK_2 +r_j + K_1$, by Lemma~\ref{lemma-KKK}, 
we have
$n_i=n_{r_i+K_1}+a_iK_{inc}$ and $n_j=n_{r_j+K_1}+a_jK_{inc}$.
 We obtain the following inequalities:
$$
\begin{array}{c}
n_{K_1}-\beta_1 + a_iK_{inc} \leq n_i \leq n_{K_1} + \beta_2 + a_iK_{inc}\\
n_{K_1}-\beta_1 + a_jK_{inc} \leq n_j \leq n_{K_1} + \beta_2 + a_jK_{inc}
\end{array}
$$
Consequently,
$$
 -\beta_1 - \beta_2 + (a_i-a_j)K_{inc} \leq n_i - n_j \leq \beta_1 + \beta_2 + (a_i-a_j)K_{inc}
$$
Considering that $(a_i-a_j) > L-1$ and using the definition of $L$, we obtain:
$$
 0 \leq \gamma K_{inc} <n_i-n_j 
$$

Hence $n_i\neq n_j$. The same proof can be done when we initially assume that  $(j-i) \geq LK_2$.\\

(2) Let us assume that $i < K_1$ and $j \geq K_1 + LK_2$. 
    Let $a_j,r_j$ be defined as for the case (1). By using the same method,
    we obtain the following inequality:

$$
n_{K_1}-\beta_1 + a_jK_{inc} \leq n_j \leq n_{K_1} + \beta_2 + a_jK_{inc}
$$

Sine $\beta_2 \geq 0$, we have:
$$
n_{K_1}-\beta_1 -\beta_2 + a_jK_{inc} -n_i \leq n_j-n_i
$$
Moreover, since $j \geq K_1 + LK_2$, we get $a_j \geq L$. Consequently,
$$
n_{K_1}-\beta_1 -\beta_2 + LK_{inc} -n_i \leq n_j-n_i
$$
Using the definition of $L$, we get
$$
n_{K_1} -\beta_1 -\beta_2 + (1  + \gamma)K_{inc} +  \beta_1 + \beta_2 - n_i 
\leq n_{K_1} -\beta_1 -\beta_2 + LK_{inc} -n_i \leq n_j-n_i
$$
Since $\gamma \times K_{inc} \geq n_i$, we get
$$
n_{K_1}  + K_{inc} \leq n_j-n_i
$$
Consequently, $n_j > n_i$.
%% Ste 081008
%%
%% If $n_i=n_j$, using the definition of $L$, we obtain~:
%% 
%% $$
%% n_{K_1}+ K_{inc}(1+\gamma)-n_i \leq 0
%% $$
%% 
%% but, since $i \in [0..K_1-1]$, we have $n_i \leq \gamma$ and $K_{inc} > 0$, we then 
%% obtain $n_{K_1} + K_{inc} \leq 0$ which is a contradiction with the fact that  $n_{K_1} \geq 0$ et $K_{inc} >0$. 
%% Hence we have $n_i \neq n_j$.
%% 

\hfill $\Box$\\
\end{proof}

Let us introduce the intermediate sets $P^1_{\sim}$ and $P^2_{\sim}$:
% Arn 110908
% $$
% P_{\sim} = \set{\pair{i}{j} \in \set{0,\ldots, K_1 + L K_2 -1}^2: 
%                 n_i = n_j}.
% $$
$$
\begin{array}{c}
P^1_\sim=\set{\pair{i}{j} \in \set{0,\ldots,K_1 + LK_2 -1}^2 \mid n_i=n_j \ {\rm and} \ i \leq j}\\
P^2_\sim=\set{\pair{i}{j} \in \set{0,\ldots,K_1 + LK_2 -1}^2 \mid n_i=n_j+LK_{inc}  \ {\rm and} \ j < i}
\end{array}
$$
In the sequel, we write $P_\sim$ to denote the set $P^1_\sim \cup P^2_\sim$.
We will now characterize the positions of $\arun_{\aautomaton}^{\omega}$ using the set $P_\sim$ and the constants $L$, $K_1$, $K_2$ and $K_{inc}$ introduced before.

\begin{lemma}\label{lemma-setP} 
Suppose $K_{inc} > 0$ and let $j \geq i$  be in $\Nat$. Then, $n_i = n_j$ iff one the conditions below
is true.

\begin{enumerate}
\itemsep 0 cm
\item $\pair{i}{j} \in P^1_\sim$.
\item $i,j \geq K_1$, $\pair{K_1 + (i-K_1) \bmod (LK_2)}{K_1 + (j-K_1) \bmod (LK_2)} \in P_\sim$ and $(j-i) < LK_2$.
\end{enumerate}
%Arn 251109
%
%%\begin{enumerate}
%%\itemsep 0 cm
%%\item $i,j <K_1 + LK_2$ and $\pair{i}{j} \in P^1_\sim$.
%%\item $i,j \geq K_1$, $\pair{K_1 + (i-K_1) mod (LK_2)}{K_1 + (j-K_1) mod (LK_2)} \in P^1_\sim$ and $(j-i) < LK_2$.
%%\item $i,j \geq K_1$, $\pair{(K_1 + (i-K_1) mod (LK_2)}{K_1 + (j-K_1) mod (LK_2)} \in P^2_\sim$ and $(j-i) < LK_2$.
%%\item $i \in \interval{0}{K_1-1}$, $K_1 \leq j$ and $\pair{i}{j} \in P^1_\sim$.
%%%\item $\set{i,j} \cap \set{0,\ldots,K_1-1}$ est un singleton, disons égal à $\set{a}$ et $\set{b}=\set{i,j} \setminus \set{a}$ et $(a,b) \in P_\sim$.
%%\end{enumerate}
%Arn 110908
% \begin{description}
% \itemsep 0 cm
% \item[(I)] $i,j  < K_1 + L K_2$ and $\pair{i}{j} \in P_{\sim}$,
% \item[(II)] $i,j \geq K_1$, 
%             $\pair{K_1 + (i-K_1) \ mod \ LK_2}{K_1+ (j-K_1) \ mod \ LK_2} \in P_{\sim}$ and 
%             $\length{i-j} < L K_2$. 
% \item[(III)] $\set{i,j} \cap \set{0,\ldots,K_1-1}$ is a singleton, say $\set{a}$ with $\set{b} = \set{i,j} 
%             \setminus \set{a}$,
%              and $\pair{a}{b}  \in P_{\sim}$. 
% \end{description}
\end{lemma}

\begin{proof}
Let $i,j \in \Nat$ such that $i \leq j$. 
If (1) is satisfied, then by definition of $P^1_\sim$, we get 
 $n_i=n_j$.

%%
%% is satisfied, which means that $i <K_1 + LK_2$ and $j <K_1 + LK_2$ et $\pair{i}{j} \in P_\sim$, then by 
%% definition of $P^1_\sim$, we have $n_i=n_j$.
%%

If (2) is satisfied, then let $r_i=(i-K_1) \bmod (LK_2)$, $r_j=(j-K_1) \bmod (LK_2)$
and $a_i$,$a_j$ be quotients such that 
$i-K_1= a_iLK_2 + r_i$ and  $j-K_1=a_jLK_2+r_j$. By Lemma~\ref{lemma-KKK}, we have $n_i=n_{r_i+K_1+a_iLK_2}=n_{r_i+K_1}+a_iLK_{inc}$ and $n_j=n_{r_j+K_1+a_j LK_2}=
n_{r_j+K_1}+a_j LK_{inc}$. Since $(j-i) < LK_2$, we have $(a_j-a_i)LK_2+(r_j-r_i) < LK_2$. Furthermore, we have by hypothesis  $\pair{K_1+r_i}{K_1+r_j} \in P_\sim$. We then distinguish two cases. First if $\pair{K_1+r_i}{K_1+r_j} \in P^1_\sim$, we deduce that $r_i \leq r_j$ and consequently $a_i=a_j$. Hence $n_i=n_j$. Second if $\pair{K_1+r_i}{K_1+r_j} \in P^2_\sim$, we deduce that  $r_j < r_i$ and consequently $a_j=a_i+1$. Hence $n_j=n_{r_j+K_1}+(a_i+1)LK_{inc}$ and since $n_{r_j+K_1}+LK_{inc}=n_{r_i+K_1}$, we obtain $n_i=n_j$.

We now suppose that $n_i=n_j$ and we perform the following case analysis.

\begin{itemize}
\itemsep 0 cm
\item Assume that $i <K_1$ and $j <K_1$. 
      By definition of $P^1_\sim$, we have $\pair{i}{j} \in P_\sim^1$ and the condition (1) is therefore
      satisfied.

\item Assume that $i,j \geq K_1$. By Lemma~\ref{lemma-bound-L}, we have $(j-i) < LK_2$ 
      (otherwise we would have $n_i \neq n_j$). 
      Let $r_i=(i-K_1) \bmod(LK_2)$, $r_j=(j-K_1) \bmod(LK_2)$ and
       $a_i$,$a_j$ be quotients such that $i-K_1= a_iLK_2 + r_i$ and  $j-K_1=a_jLK_2+r_j$.
      By Lemma~\ref{lemma-KKK}, we have 
      $n_i=n_{r_i+K_1+a_iLK_2}=n_{r_i+K_1}+a_iLK_{inc}$ and $n_j=n_{r_j+K_1+a_jLK_2}=n_{r_j+K_1}+a_jLK_{inc}$.
       We consider then two cases, according to the satisfaction of $a_i = a_j$.
\begin{itemize}
\itemsep 0 cm 
\item[\bf{--}] Suppose $a_i=a_j$. 
      Consequently, $n_{r_i+K_1}=n_{r_j+K_1}$ and since $i \leq j$, we have $r_i \leq r_j$. Condition (2) is 
      therefore satisfied.
\item[\bf{--}] Suppose $a_i \neq a_j$. Since $(j-i) < LK_2$,  
      necessarily, $a_j=a_i+1$. 
      Hence $n_{r_j+K_1}=n_i-(a_i+1)LK_{inc}$, and since $(a_j-a_i)LK_2+(r_j-r_i) < LK_2$, we also have 
       $r_j < r_i$ from which we can conclude that condition (2) is  again satisfied (we also
       have $n_{r_j+K_1} + L K_{inc} = n_{r_i+K_1}$).
\end{itemize}

\item Assume that $i <K_1$ and $j \geq K_1$. By  Lemma~\ref{lemma-bound-L}, we have $j < K_1 + LK_2$, and 
      consequently $\pair{i}{j} \in P^1_{\sim}$, hence condition (1) is satisfied.
\end{itemize}

All the values for $i,j$ are covered by the above analysis. 

\hfill $\Box$
\end{proof}

We show below how to reduce an instance of the model-checking problem (restricted to deterministic one-counter
automata) to an instance
of the problem mentioned in Theorem~\ref{theorem-without-data} by taking advantage of
Lemma~\ref{lemma-setP}. 
First let us build finite words $s,t$ over some finite alphabet $\aalphabet$.
By Lemma~\ref{lemma-purification-fo}, we can assume that the formula $\aformula$ belongs
to the pure fragment of $\fo{}{\sim,<,+1}$. 
\begin{itemize}
\itemsep 0 cm
\item $\aalphabet = \set{0, \ldots,K_1 + L K_2 -1}$.
\item $s = \set{0} \cdot \set{1} \cdots \cdot \set{K_1 - 1}$.
\item $t = \set{K_1} \cdot \set{K_1+1} \cdots \cdot \set{K_1 + L K_2 -1}$.
\end{itemize}
% Arn 110908
% Let us introduce the intermediate set $P_{\sim}$:
% $$
% P_{\sim} = \set{\pair{i}{j} \in \set{0,\ldots, K_1 + L K_2 -1}^2: 
%                 n_i = n_j}.
% $$
Given a 
sentence $\aformula$ in $\fo{}{\sim,<,+1}$ let us define a sentence 
$T(\aformula)$ in $\fo{\aalphabet}{<,+1}$ according to the 
definition below:
\begin{itemize}
\itemsep 0 cm
\item $T$ is the identity for atomic formulae of the form
       $\avariable < \avariablebis$ and $\avariable = \avariablebis + 1$.
\item $T$ is homomorphic for Boolean connectives and first-order quantification.
\item $T(\avariable \sim \avariablebis)=\big( \avariable \leq \avariablebis \wedge T_1(\avariable,\avariablebis)\big) \vee \big( \avariablebis \leq \avariable \wedge T_1(\avariablebis,\avariable)\big)$ and $T_1(\avariable,\avariablebis)$ is equal to
      $$
%% Ste 050110
%% 
%% \begin{array}{c}
%%      (\avariablebis - \avariable) <L K_2 ~\wedge \\
%%      \big(\avariable < K_1 \Rightarrow \bigvee_{\pair{I}{J} \in P^1_{\sim}} 
%% I(\avariable) \wedge J(\avariablebis) \big) \wedge \\
%%       \big(\avariable \geq K_1 \Rightarrow \bigvee_{\pair{I}{J} \in P_{\sim}} 
%% I(\avariable) \wedge J(\avariablebis) \big) 
%%       \end{array}      
%% 
(\avariablebis - \avariable) <L K_2 ~\wedge~
      \big(\avariable < K_1 \Rightarrow \bigvee_{\pair{I}{J} \in P^1_{\sim}} I(\avariable) \wedge J(\avariablebis) \big)~\wedge~
      \big(\avariable \geq K_1 \Rightarrow \bigvee_{\pair{I}{J} \in P_{\sim}} I(\avariable) \wedge J(\avariablebis) \big) 
%%Arn 261109
%%\begin{array}{c}
%%      (\avariable < K_1 + L K_2 \wedge \avariablebis < K_1 + L K_2 \wedge
%%       \bigvee_{\pair{I}{J} \in P^1_{\sim}} I(\avariable) \wedge J(\avariablebis)) \vee
%%     \\
%%      (\neg (\avariable < K_1 + L K_2 \wedge \avariablebis < K_1 + L K_2)  \wedge
%%       \bigvee_{\pair{I}{J} \in P_{\sim}} I(\avariable) \wedge J(\avariablebis) \wedge
%%\\
%%       (\avariable \geq K_1 \wedge \avariablebis \geq K_1) 
%%        \Rightarrow ((\avariablebis > \avariable \wedge (\avariablebis - \avariable) <  L K_2)
%%       )
%%\end{array}
      $$
\end{itemize}

Observe that the formula of the form $(\avariablebis - \avariable) <  L K_2$ is a shortcut for a formula
in $\fo{\locs}{<,+1}$ of polynomial size in $\length{\aautomaton}$. 
For instance, when $\avariable \geq K_1  \wedge \avariablebis \geq K_1 \wedge \avariablebis > \avariable$ holds, $(\avariablebis - \avariable) <  L K_2$ is equivalent to a formula with at most 3 variables, namely
$$
\neg \bigwedge_{I=K_1}^{K_1 + L K_2-1} \exists \ \avariableter \ \avariable \leq \avariableter < \avariablebis 
                                \wedge
                                I(\avariableter).
$$
%%Arn 261109
%%Observe that formulae of the form $\avariable < K_1 + L K_2$
%%and  $(\avariablebis - \avariable) <  L K_2$ are shortcut for formulae
%%in $\fo{\locs}{<,+1}$ of polynomial size in $\length{\aautomaton}$. 
%%By recycling variables, $\avariable < K_1 + L K_2$  can be easily obtained
%%with formulae in $\fo{\locs}{<,+1}$ with at most 3 variables. 
%%For instance, $\avariable \geq 3$ can be written as follows
%%$$
%%\exists \ \avariablebis_1 \ 
%%(
%%\avariable > \avariablebis_1 \ \wedge
%%(\exists \ \avariablebis_2 \ (\avariablebis_1 = \avariablebis_2 +1 \wedge
%%         (
%%         \exists \ \avariablebis_1 \ \avariablebis_2 = \avariablebis_1 + 1 \wedge 
%%          \neg (\exists \ \avariablebis_2 \ \avariablebis_2 < \avariablebis_1)
%%         )))
%%)
%%$$

%%Similarly, 
%%when $\avariable \geq K_1  \wedge \avariablebis \geq K_1 \wedge \avariablebis > \avariable$ holds true,
%%$(\avariablebis - \avariable) <  L K_2$ is equivalent to a formula with at most 3 variables, namely
%%$$
%%\neg \bigwedge_{I=K_1}^{K_1 + L K_2-1} \exists \ \avariableter \ \avariable \leq \avariableter < \avariablebis 
%%                                \wedge
%%                                I(\avariableter).
%%$$

\begin{lemma}  \label{lemma-into-ms}
$\aautomaton \models^{\omega} \aformula$ iff
               $s \cdot t^{\omega} \models T(\aformula)$. 
\end{lemma}

\begin{proof} The proof is by structural induction. We show that for each subformula
$\aformulabis$ of $\aformula$ and for each variable valuation $\avarval$,
$\aautomaton \models^{\omega}_{\avarval} \aformulabis$ iff
               $s \cdot t^{\omega} \models_{\avarval} T(\aformulabis)$.
Since the formula $\aformula$ belongs
to the pure fragment of $\fo{}{\sim,<,+1}$
the only case that needs to be checked is for atomic formulae of the form 
$\avariable \sim \avariablebis$. Before
giving the rest of the proof,  
we remark that since $\sigma$ is an infinite word $s \cdot t^\omega$ built over the alphabet 
$\Sigma=\set{0,\ldots,K_1 + LK_2 -1}$, for all  $i \geq K_1$, we have 
$\sigma(i)=K_1+ (i-K_1) \bmod (LK_2)$. Let $\avarval$ be a variable valuation such that 
$\avarval(\avariable)$ and $\avarval(\avariablebis)$ are defined (if $\avarval(\avariable)$ or 
$\avarval(\avariablebis)$ is not defined, then it is easy to show that $\aautomaton \not\models^{\omega}_{\avarval} \avariable 
\sim \avariablebis$ and that  $s \cdot t^{\omega} \not\models_\avarval T(\avariable \sim \avariablebis)$). 

First we suppose that $\aautomaton \models_{\avarval}^\omega \avariable \sim \avariablebis$, this means that the unique infinite accepting 
run $\arun_\aautomaton^\omega$ of $\aautomaton$ satisfies $\arun_\aautomaton^\omega \models_\avarval \avariable \sim \avariablebis$. 
Hence, we have $n_{\avarval(\avariable)}=n_{\avarval(\avariablebis)}$. We show that $s \cdot t^\omega \models_\avarval T(\avariable \sim \avariablebis)$. We suppose $\avarval(\avariable) \leq \avarval(\avariablebis)$ (the proof is similar 
for the case $\avarval(\avariablebis) \leq \avarval(\avariable)$). We proceed by a case analysis 
using Lemma~\ref{lemma-setP} and the definition for $T(\avariable \sim \avariablebis)$:

\begin{itemize}
\itemsep 0cm
\item If $\avarval(\avariable) < K_1$, then necessarily 
$(\avarval(\avariablebis)-\avarval(\avariable)) < LK_2$, hence $\sigma(\avarval(\avariable))=\avarval(\avariable)$ and $\sigma(\avarval(\avariablebis))=\avarval(\avariablebis)$, furthermore by Lemma~\ref{lemma-setP} $\pair{\avarval(\avariable)}{\avarval(\avariablebis)} \in P^1_{\sim}$, so we have $\sigma \models_\avarval T(\avariable \sim \avariablebis)$.
\item If $\avarval(\avariable) \geq K_1$, again 
we have $(\avarval(\avariablebis)-\avarval(\avariable)) < LK_2$ and also $\sigma(\avarval(\avariable))=K_1 + (i-\avarval(\avariable)) \bmod(LK_2)$ and 
$\sigma(\avarval(\avariablebis))=K_1 + (i-\avarval(\avariablebis)) \bmod(LK_2)$. Using 
Lemma~\ref{lemma-setP}, we have 
$\pair{\sigma(\avarval(\avariable))}{\sigma(\avarval(\avariablebis))} \in P_\sim$, which implies
$\sigma \models_\avarval T(\avariable \sim \avariablebis)$. 
\end{itemize}

Now, let us suppose that  $s \cdot t^\omega \models_\avarval T(\avariable \sim \avariablebis)$. 
Again, we perform a case analysis and we suppose that 
$\avarval(\avariable) \leq \avarval(\avariablebis)$ (the proof for the case $\avarval(\avariablebis) 
\leq \avarval(\avariable)$ is the same):

\begin{itemize}
\itemsep 0 cm
\item If $\avarval(\avariable) < K_1$ then $\avarval(\avariablebis) < K_1 + LK_2$. Hence 
$\sigma(\avarval(\avariable))=\avarval(\avariable)$ and $\sigma(\avarval(\avariablebis))=\avarval(\avariablebis)$. 
Since $\pair{\avarval(\avariable)}{\avarval(\avariablebis)} \in P^1_{\sim}$, we have 
$n_{\avarval(\avariable)}=n_{\avarval(\avariablebis)}$.
\item If $\avarval(\avariable) \geq K_1 $ then 
$(\avarval(\avariablebis)-\avarval(\avariable)) < LK_2$ and $\pair{\sigma(\avarval(\avariable))}{\sigma(\avarval(\avariablebis))} \in 
P_\sim$. Since $\sigma(\avarval(\avariable))=K_1 + (i-\avarval(\avariable)) \bmod(LK_2)$ and 
$\sigma(\avarval(\avariablebis))=K_1 + (i-\avarval(\avariablebis)) \bmod(LK_2)$, we obtain using Lemma \ref{lemma-setP} that $n_{\avarval(\avariable)}=n_{\avarval(\avariablebis)}$.
\end{itemize}
\hfill $\Box$
\end{proof}

This  allows us to characterize  the complexity of model checking.

\begin{theorem} \label{theorem-det} 
\ifptime
$\foMC{\omega}{}$ restricted to deterministic one-counter automata is  \pspace-com\-ple\-te
and its restriction to $n \geq 1$ variables is in \ptime.
\else
$\foMC{\infin}{}$ restricted to deterministic one-counter automata is  \pspace-com\-ple\-te.
\fi 
\end{theorem}

\begin{proof} 
Let $\aautomaton$ be a one-counter automaton and $\aformula$ be a pure formula in
$\fo{}{\sim,<,+1}$. If either $\aautomaton$ has no infinite run or its infinite run is not
accepting, then this can be checked in polynomial-time in $\length{\aautomaton}$. 
In that case $\aautomaton \models^{\omega} \aformula$ does not hold.
Moreover, observe that if $\aautomaton$ has no infinite run, then the length of the maximal
finite run is in $\mathcal{O}(\length{\locs}^3)$  by using arguments
from Lemma~\ref{lemma-KKK}. 

In the case $\aautomaton$ has an infinite accepting run and $K_{inc} > 0$,  
as shown previously the prefixes 
$s$, $t$ as well as the formula $T(\aformula)$ can be computed
in in polynomial time in $\length{\aautomaton} + \length{\aformula}$. Moreover, 
by Theorem~\ref{theorem-without-data}~\cite{Markey&Schnoebelen03},
$s \cdot t^{\omega} \models T(\aformula)$ can be checked in polynomial space in
$\length{s} + \length{t} + \length{T(\aformula)}$. 
In the case $K_{inc} = 0$, the prefixes $s$ and $t$ are defined as follows with
$\aalphabet = \set{0, \ldots,K_1 + K_2 -1}$: 
 $s = \set{0} \cdot \set{1} \cdots \cdot \set{K_1 - 1}$ and
$t = \set{K_1} \cdot \set{K_1+1} \cdots \cdot \set{K_1 + K_2 -1}$.
The map $T(\cdot)$ is defined as previouly except that
$T(\avariable \sim \avariablebis) = \bigvee_{\pair{I}{J} \in P^3_{\sim}} 
I(\avariable) \wedge J(\avariablebis)$
with $P^3_\sim=\set{\pair{i}{j} \in \set{0,\ldots,K_1 + K_2 -1}^2 \mid n_i=n_j}$. 

Hence, $\fopureMC{\infin}{}$ is in polynomial space. Using the Purification Lemma \ref{lemma-purification-fo}, we deduce that  $\foMC{\infin}{}$ is also in polynomial space.
The \pspace-hardness is a consequence of the \pspace-hardness of 
$\fMC{\omega}{}$ (since there is an obvious logspace translation from $\fLTL{\locs}{}$ into 
$\fo{\locs}{\sim,<,+1}$).
\ifptime
When the number of variables $n \geq 1$ is fixed, 
$T(\aformula)$ has at most $n + 2$ variables and since first-order model-checking with a bounded number
of variables is in \ptime, we get the \ptime \ upper bound for 
$\foMC{\omega}{n}$.
\fi
\hfill $\Box$\\
\end{proof}

\begin{theorem} \label{theorem-det-finite} 
\ifptime
$\foMC{\fin}{}$ restricted to deterministic one-counter automata is  \pspace-com\-ple\-te
and its restriction to $n \geq 1$ variables is in \ptime.
\else
$\foMC{\fin}{}$ restricted to deterministic one-counter automata is  \pspace-com\-ple\-te.
\fi 
\end{theorem}

\begin{proof} 
Let $\aautomaton$ be a one-counter automaton and $\aformula$ be a pure formula in
$\fo{}{\sim,<,+1}$. If  $\aautomaton$ has an infinite run, then
the finite words $s$ and $t$ are computed as  in the infinitary case.
We then need another intermediate set $P_F$ which will characterize the positions of the unique run labelled with an accepting state:
$$
\begin{array}{c}
P_F=\set{i \in \set{0,\ldots,K_1 + LK_2 -1} \mid \aloc_i \in F}
\end{array}
$$
The pure formula $\aformula$ is then translated into 
$$
\exists \ \avariable_{end} \ 
          (\bigvee_{I \in P_F} I(\avariable_{end})) \wedge T'(\aformula),
$$
where $T'(\aformula)$ is defined as $T(\aformula)$ for the infinitary case except
that the clause for first-order quantification becomes
$T'(\exists \ \avariable \ \aformulabis) =
\exists \ \avariable \ \avariable \leq \avariable_{end} \wedge T'(\aformulabis)$ (relativization). 
\ifptime
As in the proof of Theorem~\ref{theorem-det}, we get the \pspace \ upper bound for 
$\foMC{< \omega}{}$ and the \ptime \ upper bound for each subproblem $\foMC{< \omega}{n}$.
~~\hfill $\Box$\\
\else
As in the proof of Theorem~\ref{theorem-det}, we get the \pspace \ upper bound for 
$\foMC{\fin}{}$.
In the case $\aautomaton$ has no infinite run, then the lengh $K$ of the maximal 
finite run is in
$\mathcal{O}(\length{\locs}^3)$ and it can therefore be computed in polynomial-time. 
The prefixes $s$ and $t$ are defined as follows with
$\aalphabet = \set{0, \ldots,K - 1, \perp}$: 
 $s = \set{0} \cdot \set{1} \cdots \cdot \set{K - 1}$ and
$t = \set{\perp}$.
The map $T(\cdot)$ is defined as previouly except that
$T(\avariable \sim \avariablebis) = \bigvee_{\pair{I}{J} \in P^4_{\sim}} 
I(\avariable) \wedge J(\avariablebis)$
with $P^4_\sim=\set{\pair{i}{j} \in \set{0,\ldots,K -1}^2 \mid n_i=n_j}$. 
The pure formula $\aformula$ is  translated into 
$
\exists \ \avariable_{end} \ 
          (\bigvee_{I \in P_F'} I(\avariable_{end})) \wedge 
\neg \perp(\avariable_{end}) \wedge T'(\aformula),
$
with $P_F'=\set{i \in \set{0,\ldots,K-1} \mid \aloc_i \in F}$.
The formula $T'(\aformula)$ is defined as $T(\aformula)$ for the infinitary case except
that the clause for first-order quantification becomes
$T'(\exists \ \avariable \ \aformulabis) =
\exists \ \avariable \ \avariable \leq \avariable_{end} \wedge  T'(\aformulabis)$.
~~\hfill $\Box$\\
\fi 
\end{proof}

This improves the complexity bounds from~\cite{Demri&Lazic&Sangnier08a}. Using the translation from
$\fLTL{\downarrow}{}$  into  $\fo{}{\sim,<,+1}$ from Lemma~\ref{lemma-ltl-fo},
 we deduce the 
following corollary.

\begin{corollary} 
\ifptime
$\fMC{\infin}{}$ and $\fMC{\fin}{}$ are \pspace-complete and their 
restriction to $n \geq 1$ registers are in \ptime. 
\else
$\fMC{\fin}{}$ and $\fMC{\infin}{}$ are \pspace-complete.
\fi
\end{corollary} 

\section{Model checking nondeterministic one-counter automata}
\label{section-undecidability}

In this section, we show that 
several model-checking problems over nondeterministic one-counter automata
are undecidable by reducing decision problems for Minsky machines by following
a principle introduced in~\cite{David04}. 
Undecidability is preserved even in presence of a unique register. 
This is quite surprising  since 
$\fin$-SAT-LTL$^{\downarrow}$ restricted to one register 
and satisfiability for $\fon{}{2}(\sim,<,+1)$ 
are decidable~\cite{Bojanczyketal06a,Demri&Lazic09}.

In order to illustrate the significance of the following results, it is worth
recalling that the halting problem for
Minsky machines with incrementing errors is reducible to finitary 
satisfiability for LTL with one register~\cite{Demri&Lazic09}.  
We show below that, if we have existential model checking of one-counter automata instead of satisfiability, 
then we can use one-counter automata to refine the reduction in~\cite{Demri&Lazic09} 
so that runs with incrementing errors are excluded.  More precisely, in the reduction in~\cite{Demri&Lazic09}, 
we were not able to exclude incrementing errors because the logic is too weak to express that, 
for every decrement, 
the datum labelling it was seen before (remember that we have no past operators).  Now, the one-counter automata
are used to ensure that such faulty decrements cannot occur.

\begin{theorem} \label{theorem-undec-finitary}
$\fMC{\fin}{1}$ restricted to formulae using only the temporal operators $\mynext$ and $\sometimes$  is $\Sigma^0_1$-complete.
\end{theorem}

\begin{proof} The $\Sigma_1^0$ upper bound is by an easy verification since the existence of
a finite run (encoded in $\Nat$) verifying an $\fLTL{\downarrow,\locs}{1}$ formula
(encoded in first-order
arithmetic) can be encoded by a $\Sigma_1^0$ formula.
So, let us  reduce the halting problem for two-counter automata to 
$\fMC{\fin}{1}$ restricted to $\set{\mynext,\sometimes}$.
Let 
$ \aautomaton = \triple{\locs,\aloc_I}{\delta}{F}$ 
be a two-counter automaton: the set of instructions $L$ is 
$\set{\mathtt{inc, dec, ifzero}} \times \set{1,2}$. 
Without any loss of generality, we can assume that all the instructions
from $\aloc_I$ are incrementations. 
We  build a one-counter automaton $\aautomatonbis
 = \triple{\locs',\aloc_I'}{\delta'}{F'}$ and a sentence 
$\aformula$
in $\fLTL{\downarrow,\locs'}{1}$ such that
$\aautomaton$ reaches an accepting state iff
$\aautomatonbis \models^{\fin} \aformula$.

For each run in $\aautomaton$ of the form
$$
 \left(
      \begin{array}{c}
      \aloc_I \\
      c_1^0 = 0 \\
      c_2^0 = 0 \\
      \end{array}
    \right)
   \step{\mathtt{inst^0}}
 \left(
      \begin{array}{c}
      \aloc^1 \\
      c^1_1 \\
      c^1_2 \\
      \end{array}
    \right) 
 \step{\mathtt{inst^1}}
  \ldots
\left(
      \begin{array}{c}
      \aloc^N \\
      c^N_1 \\
      c^N_2 \\
      \end{array}
    \right) 
$$
where the $\mathtt{inst^i}$'s are instructions, we associate a  run in $\aautomatonbis$ of the form below:
$$
 \left(
      \begin{array}{c}
      \aloc_I \\
      0 
      \end{array}
    \right)
      \step{\star}
      \left(
      \begin{array}{c}
      \triple{\aloc_I}{\mathtt{inst^0}}{\aloc^1} \\
      n^1
      \end{array}
    \right)
        \step{\star}
       \left(
      \begin{array}{c}
      \triple{\aloc^1}{\mathtt{inst^1}}{\aloc^2} \\
      n^2
      \end{array}
    \right)
       \ldots
      \left(
      \begin{array}{c}
      \triple{\aloc^{N-1}}{\mathtt{inst^{N-1}}}{\aloc^N} \\
      n^N
      \end{array}
    \right)  
$$
where $\step{\star}$ hides steps for updating the counter according to the constraints described below. 
The set of states $\locs'$ will contain the set of transitions $\delta$ from $\aautomaton$. 

We first define the one-counter automaton $\aautomatonbis= \triple{\locs',\aloc_I'}{\delta'}{F'}$. 
In order to ease the presentation, 
the construction of $\aautomatonbis$ is mainly provided graphically. 
\begin{itemize}
\itemsep 0 cm
\item $\locs'$ is the following set of states:
$$
\begin{array}{ll}
\locs'= & \delta \uplus \set{\aloc_I} \uplus \set{i_0}\\
& \uplus \set{i_\atransition^{last},i_\atransition^{ \neg last} \mid \atransition=\quadruple{\aloc}{\inc}{c}{\aloc'} \in \delta}\\
& \uplus \set{d_\atransition^{last},d^{\neg last}_\atransition \mid \atransition=\quadruple{\aloc}{\dec}{c}{\aloc'} \in \delta} \\
& \uplus \set{z^{down}_\atransition \mid  \atransition=\quadruple{\aloc}{\ifzero}{c}{\aloc'} \in \delta }\\
& \uplus \set{z_\aloc \mid \aloc \in \locs} \uplus \locs_{aux} 
\end{array}  
$$
where $\locs_{aux}$ is a set of auxiliary states that we do not specify 
(but which can be identified as the states with no label in Figures~\ref{figure-inc},~\ref{figure-dec} 
and~\ref{figure-ifzero}),
\item $F'$ is the set of states  $\set{z_\aloc \mid \aloc \in F}$.
\item The transition 
      relation $\delta'$ is the smallest transition 
      relation satisfying the conditions below: 
      \begin{itemize}
      \itemsep 0 cm
        \item[\bf{--}] The transitions in Figure~\ref{fig-inc-init} belong to $\delta'$.
      \item[\bf{--}] For each incrementation transition $\atransition = 
            \triple{\aloc_I}{\mathtt{inc},c}{\aloc}$,
            the transitions in Figure~\ref{figure-inc} belong to $\delta'$.
      \item[\bf{--}] For each decrementation transition $\atransition = \triple{\aloc_I}{\dec,c}{\aloc}$, 
            the transitions in Figure~\ref{figure-dec} belong to $\delta'$.
      \item[\bf{--}] For each zero-test transition  $\atransition = \triple{\aloc_I}{\ifzero,c}{\aloc}$, the 
            transitions in Figure~\ref{figure-ifzero} belong to $\delta'$.
\end{itemize}
\end{itemize}

\begin{figure}[htbp]
\begin{center}
\begin{picture}(20,20)(0,-5)
\node(qI)(0,0){$\aloc_I$}
\node(i1)(10,10){$i_0$}
\drawedge(qI,i1){$\inc$}
\node(zq0)(20,0){$z_{\aloc_I}$}
\drawedge(i1,zq0){$\dec$}
\end{picture}
\end{center}
\caption{Initial transitions in $\delta'$}
\label{fig-inc-init}
\end{figure}
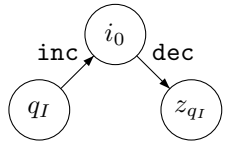

\begin{figure}[htbp]
\begin{center}
\begin{picture}(65,45)(0,-12)
\node(e0)(0,10){$z_\aloc$}
\node[Nadjust=wh](blast)(20,20){$i^{last}_{\atransition}$}
\drawedge[curvedepth=6](e0,blast){$\inc$}
\node[Nadjust=wh](bnotlast)(20,0){$i^{\neg last}_{\atransition}$}
\drawloop[loopangle=270](bnotlast){$\inc$}
\drawedge[curvedepth=-6,ELside=r](e0,bnotlast){$\inc$}
\drawedge(bnotlast,blast){$\inc$}
\node(e)(40,20){$\atransition$}
\drawedge(blast,e){$\inc$}
\node(aux)(40,0){}
\drawedge(e,aux){$\dec$}
\drawloop[loopangle=270](aux){$\dec$}
\node(z)(65,0){$z_{\aloc'}$}
\drawedge(aux,z){$\ifzero$}
\end{picture}
\end{center}
\caption{Gadget  in $\aautomatonbis$ for encoding an incrementation from $\aautomaton$}
\label{figure-inc}
\end{figure}

\begin{figure}[htbp]
\begin{center}
\begin{picture}(65,45)(0,-12)
\node(e0)(0,10){$z_\aloc$}
\node[Nadjust=wh](blast)(20,20){$d^{last}_{\atransition}$}
\drawedge[curvedepth=6](e0,blast){$\inc$}
\node[Nadjust=wh](bnotlast)(20,0){$d^{\neg last}_{\atransition}$}
\drawloop[loopangle=270](bnotlast){$\inc$}
\drawedge[curvedepth=-6,ELside=r](e0,bnotlast){$\inc$}
\drawedge(bnotlast,blast){$\inc$}
\node(e)(40,20){$\atransition$}
\drawedge(blast,e){$\inc$}
\node(aux)(40,0){}
\drawedge(e,aux){$\dec$}
\drawloop[loopangle=270](aux){$\dec$}
\node(z)(65,0){$z_{\aloc'}$}
\drawedge(aux,z){$\ifzero$}
\end{picture}
\end{center}
\caption{Gadget  in $\aautomatonbis$ for encoding a decrementation from $\aautomaton$}
\label{figure-dec}
\end{figure}

\begin{figure}[htbp]
\begin{center}
\begin{picture}(120,40)(0,5)
\node(e0)(0,10){~$z_\aloc$~}
\node(aux2)(40,10){}
\drawedge(e0,aux2){$\inc$}
\node(aux3)(20,25){}
\drawloop[loopangle=90](aux3){$\inc$}
\drawedge[curvedepth=6](e0,aux3){$\inc$}
\drawedge[curvedepth=6](aux3,aux2){$\inc$}
\node(peak)(60,10){~$\atransition$~}
\drawedge(aux2,peak){$\inc$}
\node[Nadjust=wh](tozero)(90,10){$z^{down}_{\atransition}$}
\drawloop[loopangle=90](tozero){$\dec$}
\drawedge(peak,tozero){$\dec$}
\node(e)(120,10){~$z_{\aloc'}$~}
\drawedge(tozero,e){$\ifzero$}

\end{picture}
\end{center}
\caption{Gadget  in $\aautomatonbis$ for encoding a zero-test from $\aautomaton$}
\label{figure-ifzero}
\end{figure}

In runs of $\aautomatonbis$, we are only interested in 
configurations whose state belongs to $\delta$. 
The structure of $\aautomatonbis$ ensures
that the sequence of transitions in $\aautomaton$ is valid assuming
that we ignore the intermediate (auxiliary or busy) configurations 

Before defining  the formula $\aformula$, let us introduce a few intermediate formulae
that allow us to check whether the current configuration has a state belonging
to a specific set. For each counter $i \in \set{1,2}$, we define the formulae below:
\begin{itemize}
\itemsep 0 cm
\item $I_i$ is the disjunction of $i_0$ with all the transitions
      $\atransition$
       that increment the counter
      $i$ in $\aautomaton$, hence $I_i=i_0 \vee \bigvee_{\set{\atransition \in \delta \mid 
      \atransition=\quadruple{\aloc}{\inc}{i}{\aloc'}}}\atransition$.
\item $D_i$ is the disjunction of $i_0$ with all the transitions
      $\atransition$
       that decrement the counter
      $i$ in $\aautomaton$, hence $D_i=i_0 \vee \bigvee_{\set{\atransition \in 
      \delta \mid \atransition=\quadruple{\aloc}{\dec}{i}{\aloc'}}}\atransition$.
\item $I_i^{last}$ is the disjunction of all states of the form
      $i_{\atransition}^{last}$ where $\atransition$ is a transition that increments the counter $i$, hence 
      $I_i^{last}=\bigvee_{\set{\atransition \in \delta \mid 
      \atransition=\quadruple{\aloc}{\inc}{i}{\aloc'}}}i_{\atransition}^{last}$ .
\item $I_i^{\neg last}$ is the disjunction of all states of the form
      $i_{\atransition}^{\neg last}$ where $\atransition$ is a transition that increments the counter
      $i$, hence $I_i^{\neg last}=\bigvee_{\set{\atransition \in \delta \mid \atransition=\quadruple{\aloc}{\inc}{i}{\aloc'}}}i_{\atransition}^{\neg last}$.
\item $D_i^{last}$ is the disjunction of all states of the form
      $d_{\atransition}^{last}$ where $\atransition$ is a transition that decrements the counter
      $i$, hence $D_i^{last}=\bigvee_{\set{\atransition \in \delta \mid 
      \atransition=\quadruple{\aloc}{\dec}{i}{\aloc'}}} d_{\atransition}^{last}$.
\item $D_i^{\neg last}$ is the disjunction of all states of the form
      $d_{\atransition}^{\neg last}$ where $\atransition$ is a transition that decrements the 
       counter
      $i$, hence $D_i^{\neg last}=\bigvee_{\set{\atransition \in \delta \mid \atransition=\quadruple{\aloc}{\dec}{i}{\aloc'}}}
       d_{\atransition}^{\neg last}$.
\item $Z_i$ is the disjunction of all the transitions
      $\atransition$
       that test to zero the counter
      $i$ in $\aautomaton$, hence $Z_i=\bigvee_{\set{\atransition \in \delta \mid \atransition=\quadruple{\aloc}{\ifzero}{i}{\aloc'}}}\atransition$.
\item $Z_i^{down}$ is the disjunction of  the states of the form $z_{\atransition}^{down}$ where 
      $\atransition$ is a zero-test on the counter
      $i$, hence $Z_\atransition^{down}=\bigvee_{\set{\atransition \in \delta \mid \atransition=\quadruple{\aloc}{\ifzero}{i}{\aloc'}}} z_\atransition^{down}$.
\end{itemize}

In order to define $\aformula$, we take advantage of the structure of $\aautomatonbis$ so that to 
match runs of $\aautomatonbis$ with runs of $\aautomaton$. 
%% Ste 071008
%% We will now give the formula $\aformula$
%% in $\fLTL{\downarrow,\locs'}{1}$, we will use to ``control'' the executions of $\aautomatonbis$ so that this 
%% controlled executions of $\aautomatonbis$ can be matched with the executions of the two-counter automaton 
%% $\aautomaton$.
A crucial  idea consists in associating to each action on
one of the two counters, a natural number so that an incrementation gets a new value.
Moreover, we require that the natural number associated to
an incrementation is obtained by increasing by one the natural number associated
to the previous incrementation. 
We satisfy a similar property for the natural numbers associated to decrementations
except that these values should not exceed the value associated to the previous
incrementation.  In this way, we guarantee that 
there are no more decrementations than incrementations. 
In order to simulate the zero-test,  we reach a value above all the  values that 
have been used so far. Then we check that for all the smaller values that 
are  associated to an incrementation,  it is  also associated
to a decrementation (for the same counter).

In the following formulae, we use $\always^{+}$ and $\sometimes^{+}$ to
 represent the formulae $\mynext \always$ and $\mynext \sometimes$, respectively. 
We also omit the subscript ``$1$'' in $\downarrow_1$ and $\uparrow_1$  
because we assume that we always use the same register. For each counter $i \in \set{1,2}$,
 we define the following formulae:
\begin{enumerate}
\itemsep 0cm
\item[(i)] After each configuration satisfying $I_i$, there is no strict future configuration satisfying $I_i$ with the same 
data value:
$$
\always \big( I_i \Rightarrow \downarrow \always^+ (I_i \Rightarrow \neg \uparrow )\big)
$$
\item[(ii)] After each configuration satisfying $D_i$, there is no strict future configuration satisfying $D_i$ with the same  
data value:
$$
\always \big( D_i \Rightarrow \downarrow \always^+ (D_i \Rightarrow \neg \uparrow )\big)
$$
\item[(iii)] After each configuration satisfying $D_i$, there is no strict future configuration satisfying $I_i$ with the same
data value:
$$
\always \big( D_i \Rightarrow \downarrow \always^+ (I_i \Rightarrow \neg \uparrow )\big)
$$
\item[(iv)] When a new data value is needed for an incrementation of the counter $i$, the chosen value is exactly the next value after the greatest value used so far for an incrementation of the counter $i$:
$$
\begin{array}{l}
\always \big(I_i \Rightarrow (\downarrow \sometimes(I^{\neg last}_i \wedge \uparrow) \Rightarrow \downarrow \sometimes(I^{last}_i \wedge \uparrow)) \big)\\
 \wedge \always \big( (I^{last}_i \vee I^{\neg last}_i)  \Rightarrow \downarrow \always^+ (I_i \Rightarrow \neg \uparrow) \big)
\end{array}
$$
\item[(v)] When a new data value is needed for a decrementation of the counter $i$, the chosen value is exactly the next value after the greatest value used so far for a decrementation of the counter $i$:
$$
\begin{array}{l}
\always \big(D_i \Rightarrow (\downarrow \sometimes(D^{\neg last}_i \wedge \uparrow) \Rightarrow \downarrow \sometimes(D^{last}_i \wedge \uparrow)) \big)\\
 \wedge \always \big( (D^{last}_i \vee D^{\neg last}_i)  \Rightarrow \downarrow \always^+ (D_i \Rightarrow \neg \uparrow) \big)
\end{array}
$$
\item[(vi)] The data value associated to a decrementation of the counter $i$ is never strictly greater than the greatest previous value used in incrementations of the counter $i$:
$$
\begin{array}{l}
\always \big(I_i \Rightarrow (\downarrow \sometimes(D^{\neg last}_i \wedge \uparrow) \Rightarrow \downarrow \sometimes(I^{last}_i \wedge \uparrow)) \big)\\
\wedge \always \big(I_i \Rightarrow (\downarrow \sometimes(D^{last}_i \wedge \uparrow) \Rightarrow \downarrow \sometimes(I^{last}_i \wedge \uparrow)) \big)\\
 \wedge \always \big( D^{\neg last}_i  \Rightarrow \downarrow \always^+ (I^{last}_i \Rightarrow \neg \uparrow) \big)
\end{array}
$$
\item[(vii)]  For each configuration satisfying $Z_i$, the associated data value is always strictly greater than the greatest previous value used in incrementations of the counter $i$~:
$$
\always \big( I_i \Rightarrow \downarrow \always (Z_i \Rightarrow \neg \uparrow) \big)
$$
\item[(viii)] When the automaton $\aautomatonbis$ is in the decrementing slope to encode a zero-test in 
$\aautomaton$, 
which means when the formula $Z^{down}_i$ is satisfied, and when a data value already used for an incrementation 
is met, then the same data value is used previously for a decrementation in $\aautomatonbis$:
$$
\begin{array}{l}
\neg \sometimes \big( I_i \wedge \downarrow \sometimes (Z^{down}_i \wedge \uparrow) \wedge \neg \downarrow  \sometimes (\uparrow \wedge D_i)\big)
\wedge \neg \sometimes \big ( Z^{down}_i \wedge \downarrow \sometimes (D_i \wedge \uparrow) \big)
\end{array}
$$

\end{enumerate}

Let us recall the book-keeping of the values.
\begin{itemize}
\itemsep 0 cm
\item A new  value used for an incrementation is always one plus the 
      greatest value used so far for an incrementation (see
      (iv)). The first counter value for an incrementation is 2.
\item A new  value used for decrementation is always 1 + the 
      greatest value used so far for a decrementation (see
      (v)), and is always smaller or equal to the greatest value used so for a incrementation (see (vi)). 
      The first counter value for a decrementation is 2.
\item Zero-tests consist in:
\begin{enumerate}
\itemsep 0 cm
\item[(1)] going to  a  value strictly greater than any  value used so far for incrementations (encoded in $\aautomatonbis$ and see (vii)),
      \item[(2)] then decrementing the counter to zero (encoded in $\aautomatonbis$) 
       and whenever a  value is met that is used for an incrementation, check
      that a corresponding decrementation has occured before (see (viii)). 
\end{enumerate}
\end{itemize}

In order to ease the comprehension, we explain why the rule (vi) ensures that the value associated to a 
decrementation of the counter $i$ is never strictly greater than the value used for the last incrementation of the 
same counter $i$. 
First, we assume that the rules (i)--(vi) are satisfied and {\em ad absurdum}
we suppose that the  value used for a 
decrementation is strictly greater than the  value used for the last incrementation of the counter $i$. 
If this value  is greater of exactly one unit,  then we are in the case of the second line of the 
formula given by the rule (vi).  Hence, there must exist an incrementation with the same value as the one
for the
decrementation, and this incrementation necessarily happens between the first considered incrementation and the 
decrementation, according to the rules (i)--(iii). This leads to a contradiction because the first considered 
incrementation is not the last one. Secondly, suppose that the value associated to the  decrementation is 
greater of $k$ units with $k>1$. We are  in the case of the first line of the formula given by the rule (vi),
 and consequently there exists an incrementation after the first considered incrementation which has an associated 
 value greater of one unit. The last line of the formula of the rule (vi) ensures that this incrementation 
occurs necessarily before the  decrementation, which leads again to a contradiction, because the first 
considered incrementation cannot be the last one.\\

Figure~\ref{figure-1ca} gives an example of the beginning of a run of $\aautomatonbis$ which
 respects the rules (i)--(viii)  and 
that encodes the following sequence of instructions  
$(\inc,1),(\inc,1),(\dec,1),(\dec,1),(\ifzero,1)$. 
In the decreasing part after the position labeled by $Z_1$, 
each  value used in a previous  
incrementation can be matched with a value associated to a decrementation.
\begin{figure}[htbp]
\begin{center}
\begin{picture}(120,100)(-8,-8)
\drawline(0,0)(0,90)
\drawline(0,0)(120,0)
\node[Nw=2,Nh=2](1)(0,0){}
\node[Nadjust=wh](2)(3,20){$I_1$}
\drawedge[linewidth=0.5,AHnb=0](1,2){}
\node[Nw=2,Nh=2](3)(6,0){}
\drawedge[linewidth=0.5,AHnb=0](2,3){}
\node[Nw=2,Nh=2](4)(9,20){}
\drawedge[linewidth=0.5,AHnb=0](3,4){}
\node[Nadjust=wh](5)(12,40){$I_1$}
\drawedge[linewidth=0.5,AHnb=0](4,5){}
\node[Nw=2,Nh=2](6)(15,20){}
\drawedge[linewidth=0.5,AHnb=0](5,6){}
\node[Nw=2,Nh=2](7)(18,0){}
\drawedge[linewidth=0.5,AHnb=0](6,7){}
\node[Nw=2,Nh=2](8)(21,0){}
\drawedge[linewidth=0.5,AHnb=0](7,8){}
\node[Nw=2,Nh=2](9)(24,20){}
\drawedge[linewidth=0.5,AHnb=0](8,9){}
\node[Nw=2,Nh=2](10)(27,40){}
\drawedge[linewidth=0.5,AHnb=0](9,10){}
\node[Nadjust=wh](11)(30,60){$I_1$}
\drawedge[linewidth=0.5,AHnb=0](10,11){}
\node[Nw=2,Nh=2](12)(33,40){}
\drawedge[linewidth=0.5,AHnb=0](11,12){}
\node[Nw=2,Nh=2](13)(36,20){}
\drawedge[linewidth=0.5,AHnb=0](12,13){}
\node[Nw=2,Nh=2](14)(39,0){}
\drawedge[linewidth=0.5,AHnb=0](13,14){}
\node[Nw=2,Nh=2](15)(42,0){}
\drawedge[linewidth=0.5,AHnb=0](14,15){}
\node[Nw=2,Nh=2](16)(45,20){}
\drawedge[linewidth=0.5,AHnb=0](15,16){}
\node[Nadjust=wh](17)(48,40){$D_1$}
\drawedge[linewidth=0.5,AHnb=0](16,17){}
\node[Nw=2,Nh=2](18)(51,20){}
\drawedge[linewidth=0.5,AHnb=0](17,18){}
\node[Nw=2,Nh=2](19)(54,0){}
\drawedge[linewidth=0.5,AHnb=0](18,19){}
\node[Nw=2,Nh=2](20)(57,0){}
\drawedge[linewidth=0.5,AHnb=0](19,20){}
\node[Nw=2,Nh=2](21)(60,20){}
\drawedge[linewidth=0.5,AHnb=0](20,21){}
\node[Nw=2,Nh=2](22)(63,40){}
\drawedge[linewidth=0.5,AHnb=0](21,22){}
\node[Nadjust=wh](23)(66,60){$D_1$}
\drawedge[linewidth=0.5,AHnb=0](22,23){}
\node[Nw=2,Nh=2](24)(69,40){}
\drawedge[linewidth=0.5,AHnb=0](23,24){}
\node[Nw=2,Nh=2](25)(72,20){}
\drawedge[linewidth=0.5,AHnb=0](24,25){}
\node[Nw=2,Nh=2](26)(75,0){}
\drawedge[linewidth=0.5,AHnb=0](25,26){}
\node[Nw=2,Nh=2](27)(78,0){}
\drawedge[linewidth=0.5,AHnb=0](26,27){}
\node[Nw=2,Nh=2](28)(81,20){}
\drawedge[linewidth=0.5,AHnb=0](27,28){}
\node[Nw=2,Nh=2](29)(84,40){}
\drawedge[linewidth=0.5,AHnb=0](28,29){}
\node[Nw=2,Nh=2](30)(87,60){}
\drawedge[linewidth=0.5,AHnb=0](29,30){}
\node[Nadjust=wh](31)(90,80){$Z_1$}
\drawedge[linewidth=0.5,AHnb=0](30,31){}
\node[Nw=2,Nh=2](32)(93,60){}
\drawedge[linewidth=0.5,AHnb=0](31,32){}
\node[Nw=2,Nh=2](33)(96,40){}
\drawedge[linewidth=0.5,AHnb=0](32,33){}
\node[Nw=2,Nh=2](34)(99,20){}
\drawedge[linewidth=0.5,AHnb=0](33,34){}
\node[Nw=2,Nh=2](35)(102,0){}
\drawedge[linewidth=0.5,AHnb=0](34,35){}
\put(8,-6){$(\inc,1)$}
\put(24,-6){$(\inc,1)$}
\put(42,-6){$(\dec,1)$}
\put(60,-6){$(\dec,1)$}
\put(82,-6){$(\ifzero,1)$}
\put(-18,90){Counter}
\put(-18,86){Value}
\end{picture}
\end{center}
\caption{Run for  $\aautomatonbis$ satisfying the rules (i)--(viii)}
\label{figure-1ca}
\end{figure}
The formula $\aformula$ is defined as the conjunction of (i)--(viii) plus (ix) that specifies that a state
in $F'$ is reached. 
Now consider any run of $\aautomatonbis$ which satisfies (i)--(viii). 
For any counter $c \in \set{1,2}$, 
we can define its value as the number of $I_{\atransition}$ letters
with $\atransition$ of the form
$\triple{\aloc}{\mathtt{inc},c}{\aloc'}$ for which a later letter
$\triple{\aloc_1}{\mathtt{dec},c}{\aloc'_1}$ with the same
value of the counter $\aautomatonbis$ has not yet occurred. We will now prove that $\aautomatonbis \models^{\fin} \aformula$ if and only if the automaton $\aautomaton$ has an accepting run.

Let  $\arun=\pair{\alocbis_0}{0} \step{a_0} \pair{\alocbis_1}{n_1} \step{a_1} \pair{\alocbis_2}{n_2} 
\ldots \pair{\aloc_m}{n_m}$ be a finite run of $\aautomatonbis$ satisfying 
the rules (i)--(viii) and such that $\alocbis_0=\aloc_I$ and  
$\alocbis_m=\aloc$ for some $\aloc \in \locs$.
 We consider the sequence of indices $i_1, \ldots,i_k \in \interval{0}{m}$ such that for all $j \in \interval{1}{m}$, $\alocbis_{i_j} \in \delta$ and such that there is no 
$i \in \interval{1}{m}$ with $\alocbis_i \in \delta$ and $i \not\in \set{i_1,\ldots,i_k}$. 
We will show that the sequence $p_{i_1}p_{i_2} \ldots p_{i_k}$ induces 
a run of $\aautomaton$. This means that there exist $k$ configurations 
$\aconf_1,\aconf_2, \ldots \aconf_{k} \in \locs \times \Nat^2$ such that 
$\triple{\aloc_I}{0}{0} \step{p_{i_1}} \aconf_1 \step{p_{i_2}} \aconf_2 \ldots \step{p_{i_k}} 
\aconf_k$ is a run of $\aautomaton$. 

The proof is by induction on $k$.  
If $k=1$, then by construction of the automaton $\aautomatonbis$, 
there exist $i \in \set{1,2}$ and $\aloc' \in \locs$ such that 
$p_{i_1}=\quadruple{q_0}{\inc}{i}{q'}$.
This is simply due to the fact that we have assumed that any instruction
starting in $\aloc_I$ is an incrementation.
Since it is always possible to perform an incrementation,
there is a configuration $\aconf_1 \in \locs \times \Nat^2$ such that
 $\triple{q_I}{0}{0} \step{p_{i_1}} \aconf_1$. 

We suppose that the property is true for $k$ and we show that it also holds for $k+1$.
 
First, let us  write down the properties verified by the sequence $$\pair{\alocbis_{i_0}}{n_{i_0}},\ldots,\pair{\alocbis_{i_k}}{n_{i_k}}$$ 
made of configurations of $\aautomatonbis$. For each counter $i \in \set{1,2}$, we write $Inc_i$ to denote
the set $\set{j \in \interval{1}{k} \mid p_{i_j} 
\mbox{ is of the form } \quadruple{\aloc}{\inc}{i}{q'}}$ and $Dec_i$  to denote
the set $\set{j \in \interval{1}{k} \mid \alocbis_{i_j} 
\mbox{ is of the form } \quadruple{\aloc}{\dec}{i}{q'}}$. Let $i$ be one of the counters in $\set{1,2}$. 
The rule (i) ensures that  for every $j \in Inc_i$, $n_{i_j} > 1$, 
and  for all $j,\ell \in  Inc_i$, $n_{i_j} \neq n_{i_\ell}$. This is because $i_0$ is a disjunct of $I_i$, 
the counter value in the state $i_0$ is
 always $1$ and for all $j \in Inc_i$, $p_{i_j}$ satisfies $I_i$. Furthermore the rule (iv) implies that for all $j,\ell \in  Inc_i$ 
such that $j < \ell$, if there is no $j' \in Inc_i$ such that $j < j' < \ell$, then necessarily $n_{i_\ell}=n_{i_j} +1$.
Moreover, if $j$ is the smallest index of $Inc_i$ then $n_{i_j}=2$. 
In fact, if $j$ is the smallest index of $Inc_i$, then $n_{i_j}$ is  greater or equal 
to $2$ (because the integer value in  $i_0$ is always $1$).  If  $n_{i_j}$ is strictly greater than $2$, then the run of $\aautomatonbis$ 
should reach a state that satisfies $I^{last}_i$ or $I^{\neg last}_i$ with a  value equal to $2$, but since $j$ is the smallest 
index of $Inc_i$, the rule (iv) would not be satisfied. To show the other property about the indices in $Inc_i$, 
this can be done by induction on 
 the indices of $Inc_i$ by using again the rule (iv).
Similarly, it can be proved that the set $Dec_i$ verifies the same properties. 
Hence, $\set{n_{i_j} \mid j \in Inc_i}=\interval{2}{|Inc_i|+1}$ and $\set{n_{i_j} \mid j \in Dec_i}=\interval{2}{|Dec_i|+1}$. 
Finally, the rule (vi) guarantees that for every $j \in Dec_i$, there is $\ell \in Inc_i$ such that $i_\ell \leq i_j$ and $n_{i_j} \leq n_{i_\ell}$. 
By combining these different properties, we deduce that $\length{Dec_i} \leq \length{Inc_i}$.

We suppose that $p_{i_k}=\quadruple{\aloc}{a'}{i'}{\aloc'}$. By construction of $\aautomatonbis$,
we have 
$p_{i_{k+1}}=\quadruple{\aloc'}{a}{i}{\aloc''}$. If $a$ is equal to $\inc$, then the property is satisfied because an incrementation 
can always be performed (unlike decrementations and zero-tests). 
Now, suppose that $a=\dec$. The transition $p_{i_{k+1}}=\quadruple{\aloc'}{a}{i}{\aloc''}$ is not firable only if $|Dec_i|=|Inc_i|$ 
(the number of incrementations is equal to the number of decrementations).
This situation cannot occur since $\arun$ satisfies the rules (i)---(viii), and therefore $n_{i_{k+1}}=n_{i_H}+1$ where 
$H$ is the greatest index of $|Dec_i|$ and there exists $h \in |Inc_i|$ such that $i_h \leq i_{k+1}$ and $n_{i_{k+1}} \leq n_{i_h}$. Hence,
 if $|Dec_i|=|Inc_i|$, according to the previous properties, we would have that there exists $j \in Dec_i$ such that $n_{i_h}=n_{i_j}$ and 
consequently $n_{i_H}+1 \leq n_{i_j}$ which leads to a contradiction (by definition of $H$). Now, suppose that $a=\ifzero$. The 
transition $p_{i_{k+1}}$ is not firable only if $|Inc_i|>|Dec_i|$ (there are more incrementations than decrementations).
This situation cannot occur since $\arun$ satisfies the rules (i)--(viii) and according to the rule (vii) and to the properties 
verified by $Inc_i$, for all $j \in Inc_i$, $n_{i_j} < n_{i_{k+1}}$. 
After the $i_{k+1}$th configuration, the $n_{i_{k+1}}$ next configurations contain a state that satisfies $Z^{down}_i$. If $\length{Inc_i} > 
\length{Dec_i}$ , then
this means that there is an index $h \in Inc_i$ such that for all $j \in Dec_i$, $n_{i_j} < n_{i_h}$ and there exists also  
$l \in \interval{i_{k+1}}{i_{k+1}+n_{i_{k+1}}}$ such that $\alocbis_l$ satisfies $Z^{down}_i$ and $n_l=n_h$, which is in 
contradiction with the rule (viii).

We conclude that if $\arun$ is a finite run of $\aautomatonbis$ satisfying the rules (i)--(viii) and visiting a 
state $z_{\aloc}$ in $F'$ then there is a corresponding run in the two-counter automaton $\aautomaton$ starting from the initial 
configuration $\triple{\aloc_I}{0}{0}$ and visiting the accepting state $\aloc$.

Now, we consider a run of $\aautomaton$ of the form 
$\triple{\aloc_I}{0}{0} \step{\atransition_0} \aconf_1 \step{\atransition_1} ... \step{\atransition_{h-1}} \aconf_h$. 
We show how to build a run of the one-counter automaton $\aautomatonbis$, $\pair{\alocbis_0}{0}\step{} \pair{\alocbis_1}{n_1} \step{} 
\ldots \step{} \pair{\alocbis_m}{n_m}$ with $\alocbis_0=\aloc_I$ and $\alocbis_m=z_{\aloc}$ for some $\aloc \in \locs$. 
We introduce similar notations as in the converse case.
For such a run, we consider the sequence of indices 
$i_1, \ldots,i_k \in \interval{0}{m}$ such that for all $j \in \interval{1}{m}$, $\alocbis_{i_j} \in \delta$ and such that there is no
 $i \in \interval{1}{m}$ verifying $\alocbis_i \in \delta$ and $i \not\in \set{i_1,\ldots,i_k}$. For each counter $i \in \set{1,2}$, 
we write $Inc_i$ to denote the set $\set{j \in \interval{0}{k} \mid \alocbis_{i_j} \mbox{ is of the form } \quadruple{\aloc}{\inc}{i}{\aloc'}}$
and  $Dec_i$ to denote the set $\set{j \in \interval{0}{k} \mid \alocbis_{i_j} \mbox{ is of the form } \quadruple{\aloc}{\dec}{i}{\aloc'}}$.
Finally, we define the set 
$Zero_i=\set{j \in \interval{0}{k} \mid \alocbis_{i_j} \mbox{ is of the form } \quadruple{\aloc}{\ifzero}{i}{\aloc'}}$. 
We  build a run $\arun$ of $\aautomatonbis$ such that the following properties are verified~:
\begin{enumerate}
\itemsep 0 cm
\item[(a)] $k=h$ and for all $j \in \interval{1}{k}$, $\alocbis_{i_j}= \atransition_{j-1}$,
\item[(b)] if $j$ is the smallest index of $Inc_i$, then $n_{i_j}=2$,
\item[(c)] if $j$ is the smallest index of $Dec_i$, then $n_{i_j}=2$,
\item[(d)] for all $j,\ell \in  Inc_i$ such that $j < \ell$, if there is no $j' \in Inc_i$ such that $j < j' < \ell$, then $n_{i_\ell}=n_{i_j} +1$,
\item[(e)] for all $j,\ell \in  Dec_i$ such that $j < \ell$, if there is no $j' \in Inc_i$ such that $j < j' < \ell$,then $n_{i_\ell}=n_{i_j} +1$,
\item[(f)] for all $j \in Dec_i$, there exists $\ell \in Inc_i$ such that $i_\ell < i_j$ and $n_{i_j} \leq n_{i_\ell}$,
\item[(g)] for all $j \in Zero_i$, and for all $\ell \in Inc_i$ such that $i_\ell < i_j$, we have 
           $n_{i_\ell}<n_{i_j}$ and there is $m \in Inc_i$ such that $i_m < i_j$ and $n_{i_j}=n_{i_m}+1$.
\end{enumerate}

By construction of $\aautomatonbis$, it is possible to build a run $\arun$ of $\aautomatonbis$ that satisfies the 
properties (a)--(g).
 
Now, we suppose that $\arun$ is a run of $\aautomatonbis$ verifying these properties and 
it remains to check that
 $\arun$ satisfies the rules (i)--(viii). First, we consider the rules (i)--(ii). These two rules are satisfied because all the 
elements of $Inc_i$ and of $Dec_i$ are built with distinct values for incrementations and
decrementations. The rule (iii) is satisfied because of the properties (e) and (f). 
The rule (iv) is satisfied, 
because if the run is in a position $i_j$ with $j \in Inc_i$ and if there exists a position $\ell$ in the future which satisfies $I^{\neg last}_i$, then 
there exists a position $i_{j'}$ such that $ \ell < i_{j'}$ with $j' \in Inc_i$ and $n_{i_{j'}} > n_{i_j} + 1$ (by construction of $\aautomatonbis$ and
by (d)). Moreover, the definition of $\aautomatonbis$ implies there exists a position $h$ such that $i_j < h < \ell$, $h$ satisfies 
$I^{last}_i$, $n_h=n_{i_j}$, $q_{h+1}$ satisfies $I_i$  and $n_{h+1}=n_{i_j}+1$ . 
Similar arguments are used to establish that the rule (v) is satisfied by using  (c) and (e). 
The rule (vi) is satisfied because of the property (f).  Finally the 
rules (vii)--(viii) 
are satisfied by using  (g) and the properties about the sets $Inc_i$ and $Dec_i$. Hence 
if there is a run of $\aautomaton$ leaving from $\triple{q_I}{0}{0}$ and visiting a state $\aloc$ in $F$, we can build a finite run
 $\arun$ of $\aautomatonbis$ such that $\arun \models \aformula$.

Furthermore the formula $\aformula$ uses only the temporal operators $\mynext$ and $\sometimes$ 
(the operator $\always$ can be easily obtained from $\sometimes$).
\hfill $\Box$

% Arn 1609098
% For any run of $\aautomatonbis$ which satisfies (i)-(viii), we can thus
% extract a valid run of $\aautomaton$. Conversely, any valid run of 
% $\aautomaton$ can be encoded in the same way as a run of $\aautomatonbis$
% which satisfies (i)-(viii).
% The latter is done by inserting auxiliary letters as required to reach
% appropriate values of the counter of $\aautomatonbis$. 
% \fbox{{\sc TO BE COMPLETED}}
%  \qed
\end{proof}

%% 
%% {\it stephane:
%% In the above proof, no zero-test instruction is performed. 
%% If we allow the guard $\top$ (true) in transitions (keeping unchanged the counter), we can encoded
%% the value of the current location $\aloc$ by an adequate sequence of $\top$ transitions (of length $n_{\aloc}$).
%5 Indeed, $\top$ transitions cannot belong to the initial automaton.
%% However, I do not see how to show the undecidability of  
%% $\fpureMC{< \omega}{1}{\sim}$ without $\top$ transitions.
%% }
%% 

\begin{theorem} $\fMC{\omega}{1}$ restricted to $\set{\mynext,\sometimes}$ is $\Sigma^1_1$-complete.
\end{theorem}

The proof is similar to the proof of Theorem~\ref{theorem-undec-finitary}
except that instead of reducing the halting problem for Minsky machines, we 
reduce the recurrence problem for nondeterministic Minsky machines
that is known to be $\Sigma^1_1$-hard~\cite{Alur&Henzinger89}. The $\Sigma_1^1$ upper bound is by an easy verification
since an accepting run can be viewed as a function $\amap:\Nat \rightarrow \Nat$ and then checking
that it satisfies an $\fLTL{\downarrow,\locs}{1}$ formula can be expressed in first-order arithmetic. 
Another consequence of the Purification Lemma is the result below.

\begin{theorem} $\fpureMC{\fin}{1}$ restricted to $\set{\mynext,\sometimes}$ 
 is $\Sigma^0_1$-complete.
 $\fpureMC{\infin}{1}$ restricted to $\set{\mynext,\sometimes}$  is $\Sigma^1_1$-complete. 
\end{theorem}

This refines results stated in~\cite{Demri&Lazic&Sangnier08a}.

Using Theorem 3.2(a) in~\cite{Demri&Lazic09}, we can obtain the following corollary by a direct analysis 
of the formulae involved in the proof of Theorem~\ref{theorem-undec-finitary} (every temporal operator is
prefixed by a freeze operator or can occur equivalently in such a form).

\begin{corollary} $\fotwoMC{\fin}{}$ [resp.\ $\fotwoMC{\infin}{}$] without the predicate $+1$
is $\Sigma^0_1$-complete [resp.\ $\Sigma^1_1$-complete] and
$\fopureMC{\fin}{4}$ [resp.\ $\fopureMC{\infin}{4}$] 
is $\Sigma^0_1$-complete [resp.\ $\Sigma^1_1$-complete].
\end{corollary}

The absence of the predicate $+1$ in the above corollary is due to the fact that in the proof
of Theorem~\ref{theorem-undec-finitary}, $\mynext$ occurs only to encode $\sometimes^+$ and $\always^+$. 
The above-mentioned undecidability is true even if we restrict ourselves 
to one-counter automata for which there are no transitions
with identical instructions leaving from the same state. 
A one-counter automaton $\aautomaton$ is \emph{weakly deterministic} whenever 
for every state $\aloc$, if $\triple{\aloc}{l}{\aloc'}, \triple{\aloc}{l'}{\aloc''}
 \in \delta$, we have $l = l'$ implies $\aloc' = \aloc''$. 
The transition systems induced by these  automata
are not necessarily deterministic.

\begin{theorem} \label{theorem-undecidability-weakly}
$\fpureMC{\fin}{1}$ [resp.\ $\fpureMC{\infin}{1}$] restricted to weakly
deterministic one-counter automata is $\Sigma^0_1$-complete [resp.\ $\Sigma^1_1$-complete]. 
\end{theorem}

\iffossacs
The proof uses the Purification Lemma and provides reductions from the model-checking
problems to their restrictions to weakly deterministic automata. 
\else
\begin{proof}
In the proof of the Purification Lemma, weak determinisn of the one-counter automata
is preserved. It is sufficient to show that 
given a one-counter automaton $\aautomaton$
              and  a sentence $\aformula$ in 
              $\fLTL{\downarrow,\locs}{}$, one can compute 
%%Arn 261109
%%
%%              in logarithmic space in
%%$\length{\aautomaton} + \length{\aformula}$ 
a weakly deterministic automaton 
$\aautomaton'$ and $\aformula'$ in $\fLTL{\downarrow,\locs'}{}$ ($\locs \subseteq \locs'$) 
such that  $\aautomaton \models^{\fin} \aformula$
                 [resp. $\aautomaton \models^{\infin} \aformula$ ] iff 
$\aautomaton' \models^{\fin} \aformula'$
                 [resp. $\aautomaton' \models^{\infin} \aformula'$].

Figure~\ref{fig:weak-det} illustrates with examples how transitions from a state
with identical instructions can be transformed so that to 
obtain a weakly deterministic automaton.
In Figure~\ref{fig:weak-det}, 
we have omitted the transitions labelled by a zero-test or a decrementation when 
they are never fired. 
This can be easily generalized to all the transitions of $\aautomaton$. 
The formula $\aformula'$ is defined as $\atranslation(\aformula)$ with 
the map $\atranslation$ that is homomorphic for Boolean operators and $\downarrow_r$,
and its restriction to atomic formulae is identity. 
It remains to define the map for the temporal operators, which corresponds to 
perform a relativization:
\begin{itemize}
\itemsep 0 cm
\item $\atranslation(\aformula_1 \until \aformula_2)=
      \big((\bigvee_{\aloc \in \locs} \aloc) 
           \Rightarrow \atranslation(\aformula_1)\big) \until 
      \big(\bigvee_{\aloc \in \locs} \aloc \wedge \atranslation(\aformula_2)\big)$,
\item $\atranslation(\mynext \aformulabis )= \mynext 
         \big((\neg \bigvee_{\aloc \in \locs} \aloc) \ \until \ 
         (\bigvee_{\aloc \in \locs} \aloc \wedge \atranslation(\aformulabis))\big)$.
\end{itemize}

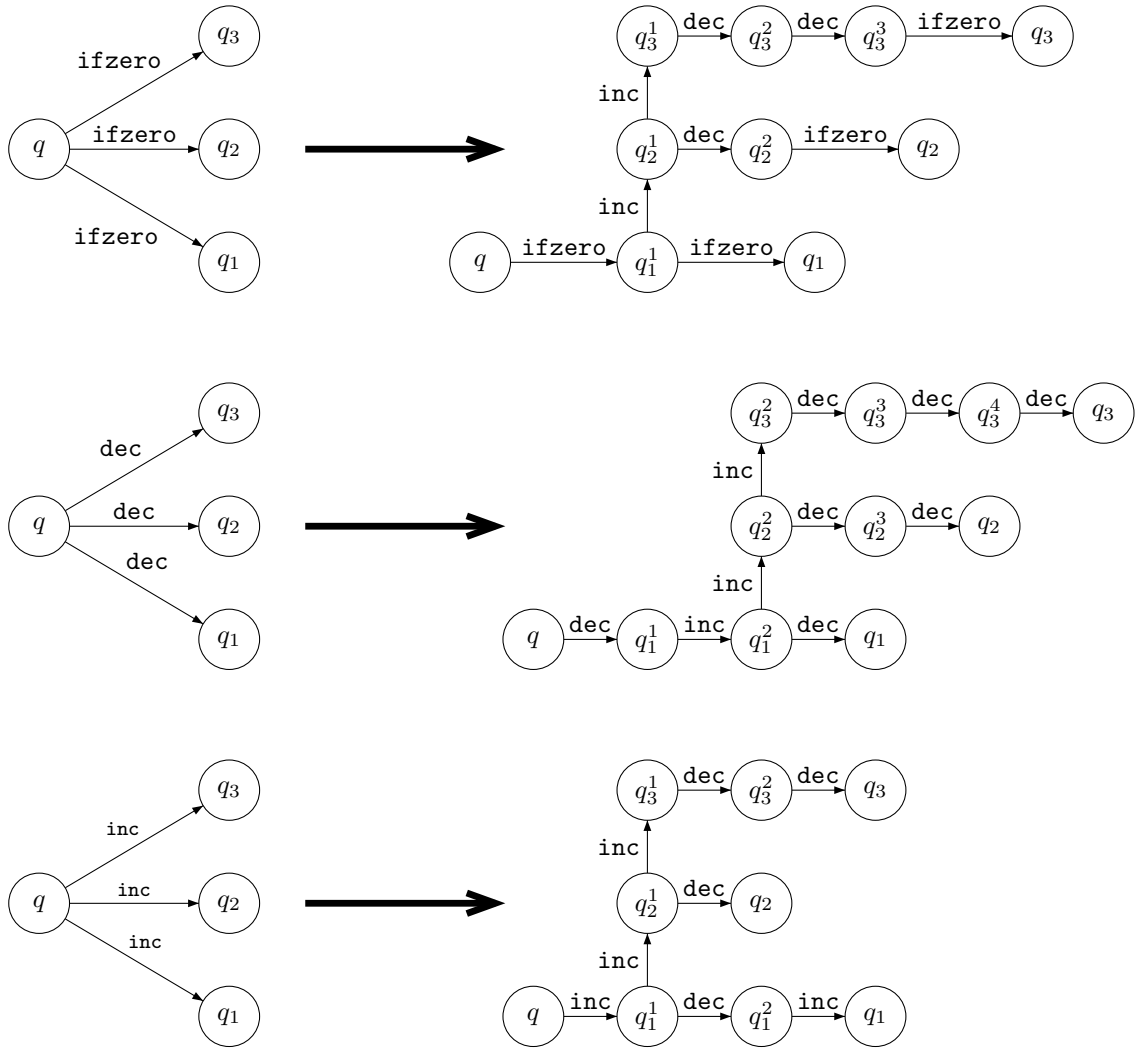
\begin{figure}[htbp]
\begin{picture}(145,140)(-10,0)
%
%\unitlength=1mm
%
\node(Q)(0,20){$\aloc$}
\node(Q1)(25,5){$\aloc_1$}
\node(Q2)(25,20){$\aloc_2$}
\node(Q3)(25,35){$\aloc_3$}
\drawedge(Q,Q1){\footnotesize $\mathtt{inc}$}
\drawedge(Q,Q2){\footnotesize $\mathtt{inc}$}
\drawedge(Q,Q3){\footnotesize $\mathtt{inc}$}
%
%\node(L1)(35,20){}
%\node(L2)(55,20){}
%\drawedge(L1,L2){}
%\drawedge[NHab=0](L1,L2){}

\drawline[linewidth=0.9,AHLength=4,AHlength=0,AHnb=1](35,20)(60,20)

\node(WQ)(65,5){$\aloc$}
\node(A)(80,5){$\aloc_1^1$}
\drawedge(WQ,A){$\mathtt{inc}$}
\node(B)(95,5){$\aloc_1^2$}
\drawedge(A,B){$\mathtt{dec}$}
\node(C)(110,5){$\aloc_1$}
\drawedge(B,C){$\mathtt{inc}$}
\node(D)(80,20){$\aloc_2^1$}
\drawedge(A,D){$\mathtt{inc}$}
\node(E)(95,20){$\aloc_2$}
\drawedge(D,E){$\mathtt{dec}$}
\node(F)(80,35){$\aloc_3^1$}
\drawedge(D,F){$\mathtt{inc}$}
\node(G)(95,35){$\aloc_3^2$}
\drawedge(F,G){$\mathtt{dec}$}
\node(H)(110,35){$\aloc_3$}
\drawedge(G,H){$\mathtt{dec}$}
\node(DQ)(0,70){$\aloc$}
\node(DQ1)(25,55){$\aloc_1$}
\node(DQ2)(25,70){$\aloc_2$}
\node(DQ3)(25,85){$\aloc_3$}
\drawline[linewidth=0.9,AHLength=4,AHlength=0,AHnb=1](35,70)(60,70)
\drawedge(DQ,DQ1){$\mathtt{dec}$}
\drawedge(DQ,DQ2){$\mathtt{dec}$}
\drawedge(DQ,DQ3){$\mathtt{dec}$}
\node(DWQ)(65,55){$\aloc$}
\node(DA)(80,55){$\aloc^1_1$}
\drawedge(DWQ,DA){$\mathtt{dec}$}
\node(DB2)(95,55){$\aloc^2_1$}
\drawedge(DA,DB2){$\mathtt{inc}$}
\node(DB3)(110,55){$\aloc_1$}
\drawedge(DB2,DB3){$\mathtt{dec}$}
\node(DD)(95,70){$\aloc_2^2$}
\drawedge(DB2,DD){$\mathtt{inc}$}
\node(DE)(110,70){$\aloc_2^3$}
\drawedge(DD,DE){$\mathtt{dec}$}
\node(DB)(125,70){$\aloc_2$}
\drawedge(DE,DB){$\mathtt{dec}$}
\node(DF)(95,85){$\aloc_3^2$}
\drawedge(DD,DF){$\mathtt{inc}$}
\node(DG)(110,85){$\aloc_3^3$}
\drawedge(DF,DG){$\mathtt{dec}$}
\node(DH)(125,85){$\aloc_3^4$}
\drawedge(DG,DH){$\mathtt{dec}$}
\node(DI)(140,85){$\aloc_3$}
\drawedge(DH,DI){$\mathtt{dec}$}
\node(ZQ)(0,120){$\aloc$}
\node(ZQ1)(25,105){$\aloc_1$}
\node(ZQ2)(25,120){$\aloc_2$}
\node(ZQ3)(25,135){$\aloc_3$}
\drawline[linewidth=0.9,AHLength=4,AHlength=0,AHnb=1](35,120)(60,120)
\drawedge[ELside=r](ZQ,ZQ1){$\mathtt{ifzero}$}
\drawedge(ZQ,ZQ2){$\mathtt{ifzero}$}
\drawedge(ZQ,ZQ3){ $\mathtt{ifzero}$}
\node(ZWQ)(58,105){$\aloc$}
\node(ZA)(80,105){$\aloc^1_1$}
\drawedge(ZWQ,ZA){$\mathtt{ifzero}$}
\node(ZB2)(102,105){$\aloc_1$}
\drawedge(ZA,ZB2){$\mathtt{ifzero}$}

\node(ZD)(80,120){$\aloc_2^1$}
\drawedge(ZA,ZD){ $\mathtt{inc}$}
\node(ZE)(95,120){$\aloc_2^2$}
\drawedge(ZD,ZE){$\mathtt{dec}$}
\node(ZB)(117,120){$\aloc_2$}
\drawedge(ZE,ZB){$\mathtt{ifzero}$}
\node(ZF)(80,135){$\aloc_3^1$}
\drawedge(ZD,ZF){$\mathtt{inc}$}
\node(ZG)(95,135){$\aloc_3^2$}
\drawedge(ZF,ZG){$\mathtt{dec}$}
\node(ZH)(110,135){$\aloc_3^3$}
\drawedge(ZG,ZH){$\mathtt{dec}$}
\node(ZI)(132,135){$\aloc_3$}
\drawedge(ZH,ZI){$\mathtt{ifzero}$}
\end{picture}
\caption{Weak determinization of one-counter automata}
\label{fig:weak-det}
\end{figure}

It can be easily proved that $\aautomaton'$ and $\aformula'$ satisfy the desired properties.
\qed
\end{proof}
\fi

 \section{Conclusion}
\label{section-conclusion}

In the paper, we have studied  complexity issues
related to the model-checking problem for LTL with registers
over one-counter automata. Our results are quite different from those for satisfiability. 
We have shown that model checking  $\fLTL{\downarrow}{}$ restricted to the operators
$\set{\mynext,\sometimes}$ and $\fon{}{2}(\sim,<,+1)$ over one-counter automata
is undecidable, which contrasts
with the decidability of many verification problems for one-counter 
automata~\cite{Jancaretal04,Serre06,Demri&Gascon07} and with the  results
in~\cite{Bojanczyketal06a,Demri&Lazic09}. 
For instance, we have shown that 
model checking nondeterministic one-counter automata over $\fLTL{\downarrow}{}$ restricted
      to a unique register and without alphabet 
[resp. $\fon{}{2}(\sim,<,+1)$] is already $\Sigma_1^1$-complete in the infinitary case.
On the decidability side, 
the \pspace \ upper bound for model checking $\fLTL{\downarrow}{}$ and  $\fo{}{\sim,<,+1}$ 
over deterministic one-counter automata in the infinitary and finitary cases 
is established by using in an essential
way~\cite{Markey&Schnoebelen03} (and simplifying the proofs from~\cite{Demri&Lazic&Sangnier08a}).
In particular, we have established that the runs of deterministic one-counter automata admit descriptions
that require polynomial size only. 
%Figure~\ref{figure-table} contains a summary of the main results. 
Hence, our results essentially deal with LTL with registers but they can be also understood as a contribution
to  refine the decidability border for problems on one-counter automata. 

%\begin{figure}
%{\small 
% \begin{center}
%\begin{tabular}{|c|c|c|}
%\hline
%\pspace-completeness  & $\Sigma_1^0$-completeness & $\Sigma_1^1$-completeness \\
% for det. 1CA & for 1CA & for 1CA \\
%\hline
% $\fMC{\omega}{}$, $\fMC{< \omega}{}$  &  $\fMC{< \omega}{1}[\mynext,\sometimes]$ & 
%      $\fMC{\omega}{1}[\mynext,\sometimes]$ \\
% & $\fpureMC{< \omega}{1}[\mynext,\sometimes]$ & $\fpureMC{\omega}{1}[\mynext,\sometimes]$ \\ 
%\hline
%$\foMC{\omega}{}$, $\foMC{< \omega}{}$ & $\fotwoMC{< \omega}{}[\sim,<]$ & $\fotwoMC{\omega}{}[\sim,<]$ \\
%\hline
%\end{tabular}
%\end{center}
%} 
%\caption{Summary of main results}
%\label{figure-table}
%\end{figure}

\iffossacs
Viewing runs as data words is an idea that can be pushed further.
For instance, the decidability status
of model checking  $\fLTL{\downarrow}{}$ over the class of 
reversal-bounded counter automata~\cite{Dang&Ibarra&SanPietro01}
 remains  open. Hence, 
our results pave the way
for  model checking   $\fLTL{\downarrow}{}$
over other classes of operational models that are known to admit powerful techniques
for solving verification tasks.
\else
Viewing runs as data words is an idea that can be pushed further.
Indeed, our results pave the way
for  model checking   memoryful (linear-time) logics (possibly extended to  multicounters)
over other classes of operational models that are known to admit powerful techniques
for solving verification tasks.
For instance, the reachability relation is known
to be Presburger-definable for reversal-bounded counter automata~\cite{Ibarra78}.
Nevertheless, model checking  $\fLTL{\downarrow}{}$ over this class of counter machines has been recently
shown undecidable~\cite{Demri&Sangnier10}; other subclasses of counter machines for which the reachability
problem is decidable have been considered in this recent work. 
\fi

{\em Acknowledgement:} We would like to thank Philippe 
Schnoebelen for 
suggesting simplifications in  the proofs of 
Lemma~\ref{lemma-purification} and Proposition~\ref{proposition-pspace-hardness}
and Luc Segoufin for fruitful discussions that lead us to improve significantly the results
from~\cite{Demri&Lazic&Sangnier08a}.

%% \clearpage
%% \nocite{*}
\bibliographystyle{elsarticle-num}
\bibliography{biblio}

\end{document}